\documentclass[twocolumn,journal]{IEEEtran}
\usepackage{amsmath,amsthm,amssymb,mathtools,graphicx,cite,calc,color,dsfont,psfrag}
\usepackage{array}
\usepackage{hyperref}
\usepackage{multirow}
\usepackage{pgf,tikz}
\usepackage{amsfonts}
\usepackage{bm}
\usepackage{subfigure}
\usepackage{hyperref}
\usepackage{enumitem}

\newtheorem{theo}{Theorem}

\theoremstyle{definition}
\newtheorem{definition}{Definition}

\newtheorem{rmk}{Remark}
\newcommand{\EX}{{\mathbb{E}}} 
\newcommand{\IDM}{{\mathbf{I}}} 
\newcommand{\TRAC}{{\mathrm{tr}}} 
\newcommand{\MY}{{\bm Y}}

\newcommand{\MSIG}{{\bm \Sigma}}

\title{Large Array Antenna Spectrum Sensing in Cognitive Radio Networks}

\author{\IEEEauthorblockN{Amirhossein Taherpour$^{*}$, Abbas Taherpour$^{**}$, and Tamer Khattab$^{\dag}$}\\
\IEEEauthorblockA{$^{*}$ Dept. of EE, Columbia University, New York, NY, USA \\
$^{**}$ Dept. of EE, Imam Khomeini International University, Qazvin, Iran\\
$^{\dag}$ Dept of EE, College of Engineering, Qatar University, Doha, Qatar
}}
\begin{document}
\maketitle
\thispagestyle{empty}
\date{}

\begin{abstract}
We investigate the problem of spectrum sensing in cognitive radios (CRs) when the receivers are equipped with a large array of antennas. We propose and derive three detectors based on the concept of linear spectral statistics (LSS) in the field of random matrix theory (RMT). These detectors correspond to the generalized likelihood ratio (GLR), Frobenius norm, and Rao tests employed in conventional multiple antenna spectrum sensing (MASS). Subsequently, we compute the Gaussian distribution of the proposed detectors under the noise-only hypothesis, leveraging the central limit theorem (CLT) applied to high-dimensional random matrices. We evaluate the performance of the proposed detectors and analyze the impact of the number of antennas and samples on their efficacy. Furthermore, we assess the accuracy of the theoretical results by comparing them with simulation outcomes. The simulation results provide evidence that the proposed detectors exhibit efficient performance in wireless networks featuring large array antennas. These detectors find practical applications in diverse domains, including massive MIMO wireless communications, radar systems, and astronomical applications.
\end{abstract}


\section{Introduction} \label{sec:intro}
\IEEEPARstart
Accurate and prompt spectrum sensing plays a crucial role in the functionality of cognitive radio networks (CRNs)\cite{Ahmad2020,taherpour2010,shen2012,ramirez2011detection,hanafi2013,tugnait2012}. During spectrum sensing, secondary users (SUs) must efficiently and quickly determine the presence or absence of primary users (PUs) despite wireless channel impairments such as noise and fading. Multiple antenna spectrum sensing (MASS) is an effective technique that leverages the inherent spatial diversity of wireless channels through the use of antenna arrays to combat fading and other impairments. Consequently, numerous MASS techniques have been proposed in the literature to capitalize on its potential benefits \cite{orooji2011,larsonICASSP,taherpour2013,pour2015,leshem2001,multipleprimaryuser,baysian2013,zhang2020}. MASS detectors have been extensively studied for various scenarios, including scenarios involving correlated and independent antennas, the rank of the PU's covariance matrix, and the availability of prior information on system parameters \cite{orooji2011,larsonICASSP,taherpour2013,pour2015,leshem2001,multipleprimaryuser,baysian2013}. Most of these spectrum sensing techniques rely on eigenvalue-based detectors, which are derived as solutions to composite hypothesis testing problems using generalized likelihood ratio (GLR), Wald, and Rao tests \cite{kaybook}. Furthermore, it has been demonstrated that higher-order moments of eigenvalues of the sample covariance matrix (SCM) can be employed for more accurate spectrum sensing under certain conditions \cite{taherpour2017}.

Until recently, studies and performance evaluations in the field of MASS have primarily focused on systems with a maximum of eight antennas, in line with the long-term evolution (LTE) wireless network standards. Practical multiple-input and multiple-output (MIMO) systems commonly employ 4 or 2 antennas \cite{densification2, disantenna}. However, with the emergence of fifth-generation (5G) wireless network systems, massive MIMO systems utilizing large array antennas have gained significant attention. This is because notable advancements have been made in implementing large array antennas, even in constrained environments \cite{5gwillbe, massivemimo, scalemimo}. These systems comprise a substantial number of antennas, often reaching several hundred, enabling the realization of fixed and mobile broadband systems with improved energy and spectral efficiency \cite{scalemimo,pucci2022system}.

In conventional MASS scenarios, where the number of samples exceeds the number of antennas, the likelihood ratio test (LRT) is commonly used as the test statistic. The LRT and its variants, such as the generalized likelihood ratio test (GLRT), rely on estimating covariance matrices from the SCM and utilize the eigenvalues of the SCM. However, in the context of large array antennas, where the number of antennas can be significant and comparable to the number of samples, the LRT is no longer applicable. This is due to the challenge of accurately estimating the covariance matrix based on the SCM, which becomes singular when the number of antennas exceeds the number of samples. As a result, spectrum sensing detectors based on eigenvalues become unreliable and ineffective in such scenarios \cite{Baibook}.

Another distinctive characteristic of MASS for large array antennas pertains to the asymptotic distributions of statistical tests. While conventional detection and estimation theory establish that, for a given number of antennas, the distributions of LRTs follow a chi-square distribution as the number of samples tends to infinity \cite{kaybook}, this is no longer valid for large array antennas, where the number of antennas may surpass the number of samples. As we delve into the realm of random matrix theory (RMT) and explore the central limit theorem (CLT) for random matrices, we uncover that the distributions of derived detectors exhibit Gaussian characteristics with specific means and variances. Nonetheless, computing these means and variances necessitates complex and intricate integrations \cite{CLT2004,CLT2009}.

In light of the aforementioned limitations surrounding the application of conventional LRTs in large array antenna spectrum sensing, alternative approaches must be explored to devise effective tests in such scenarios. RMT emerges as a powerful framework to tackle this issue and address similar problems where the dimensions of the relevant matrices are large \cite{bianchi2011}. When dealing with random matrices, if the dimension of the matrix exceeds a certain threshold, one can assume that the matrix dimensions are extremely large and employ RMT to obtain insightful results. Remarkably, the approximation methods employed in such cases exhibit remarkable accuracy and can yield valuable outcomes even for matrices as small as $8\times8$ or, in some cases, $4\times4$ \cite{romain}.

This paper explores the problem of MASS using RMT in scenarios where the SU is equipped with a large receiver antenna array. Building upon existing RMT results and assuming a specific form of linear spectral statistics (LSS), we derive three detectors. These detectors are extensions of conventional MASS techniques, including the GLRT, Frobenius norm, and Rao tests, specifically designed for large array antennas. Moreover, we establish that the proposed MASS detectors exhibit Gaussian distributions and determine their means and variances using a specialized form of the Central Limit Theorem (CLT) for random matrices.

The remaining sections of the paper are organized as follows. Section \ref{sec:formul} introduces the system model, basic assumptions, and key results from RMT that will be utilized in subsequent sections to derive the proposed detectors and assess their performance. In Section \ref{sec:detectors}, we derive three detectors specifically designed for large array antennas based on the LSS function and relevant theorems. The distribution of the proposed detectors under the noise-only hypothesis is calculated using a specialized version of the CLT in RMT in Section \ref{sec:pfa}. Simulation results and performance evaluations are presented in Section \ref{sec:simulation}. Finally, Section \ref{sec:conclusion} provides concluding remarks to summarize the paper.
\section{Problem Formulation and Random Matrix Theory Preliminaries}\label{sec:formul}
In this section, we provide the system model and assumptions, and then present an overview of the preliminary results obtained from RMT. These results will serve as the basis for deriving the tests specifically designed for large array antennas and evaluating their performance in the upcoming sections.

\subsection{System Model}\label{subsec:model}
Consider the scenario where the received radio frequency (RF) signal is down-converted and sampled at the Nyquist rate, resulting in $L$ complex temporal samples acquired at each of the $M$ antennas. The received samples at the $\ell$-th time instant, under hypotheses $\mathcal{H}_0$ and $\mathcal{H}_1$ representing the absence and presence of the primary user (PU), respectively, can be expressed as:

\begin{align}
\mathbf{y}_\ell =
\begin{cases}
\mathbf{n}_\ell,& \mathcal{H}_0\\
\mathbf{h}_\ell\mathbf{s}_\ell+\mathbf{n}_\ell,&\mathcal{H}_1
\label{hyp1}
\end{cases}
&; ~~~~\ell=1,2,\dots,L,
\end{align}

where $\mathbf{y}_{\ell}\doteq[\mathbf{y}_{\ell}^{(1)},\dots,\mathbf{y}_{\ell}^{(M)}]^t\in\mathbb{C}^M$ and $\mathbf{n}_\ell\doteq[\mathbf{n}_{\ell}^{(1)},\dots,\mathbf{n}_{\ell}^{(M)}]^t\in\mathbb{C}^M$, for $\ell=0,\dots,L$, represent the column vector of received complex baseband signals and the zero-mean circular complex additive white Gaussian noise (AWGN) with covariance matrix $\mathbf{\Sigma}_N$ at the $M$ antennas. Furthermore,$\bm{s}_{\ell} \in \mathbb{C}$ denotes zero-mean independent and identically distributed (\emph{i.i.d.}) samples of the PU's signal vector, denoted as $\bm{s}$, at the $\ell$-th time instant with variance $\mathcal{E}_s$. We assume that the samples of the PU's signal follow a Gaussian distribution, which holds true for constant envelope modulated signals and orthogonal frequency division multiplexing (OFDM) with a large number of carriers \cite{pour2015}. The channel gain vector between the PU and the $M$ antennas at the $\ell$-th time instant is modeled by the vector $\mathbf{h}_\ell$, assumed to have a zero-mean circular complex Gaussian distribution with a covariance matrix $\mathbf{\Sigma}_H$. We also assume that the sensing bandwidth is comparable to the PU's bandwidth, resulting in \emph{i.i.d.} channel gains over the sensing times, which remain fixed during the sensing period as a realization of the random variable.

Let $\mathbf{Y} \doteq [\mathbf{y}_{1}, \dots, \mathbf{y}_{L}] \in \mathbb{C}^{M \times L}$ represent the space-time observation matrix, constructed by concatenating the vectors $\mathbf{y}_{\ell}$ as $L$ columns of the matrix. Under the above assumptions, the observation matrix $\mathbf{Y}$ consists of independent zero-mean rows, resulting from the temporally uncorrelated samples at each antenna. Now, for $1 \le \ell \le L$, we can rewrite (\ref{hyp1}) as: 
\begin{flalign}
\begin{cases}
\mathcal{H}_0:\mathrm{col}_\ell[\MY]\thicksim\mathcal{CN}(\mathbf{0},\bm{\Sigma}_0)&,\quad
\text{PU is absent}\\
\mathcal{H}_1:\mathrm{col}_\ell[\MY]\thicksim\mathcal{CN}(\mathbf{0},\bm{\Sigma}_1)&,\quad
\text{PU is present}
\label{eq2}
\end{cases}
\end{flalign}
where $\bm{\Sigma}_\nu$ for $\nu=0,1$ is the covariance matrix of $\ell$-th column of $\bm{Y}$ under hypothesis $\mathcal{H}_\eta$ and in general can be expressed as follows
\begin{align}
\mathbf{\Sigma}_0=\mathbf{\Sigma}_N~~,~~\mathbf{\Sigma}_1=\mathcal{E}_s\mathbf{\Sigma}_H+\mathbf{\Sigma}_N.
\end{align}

As we consider a complex additive white Gaussian noise (AWGN) model, the covariance matrix $\mathbf{\Sigma}_N$ is assumed to be a diagonal matrix. In the calibrated case, all diagonal entries are the same, while in the uncalibrated case, they may differ.

We define the SCM as $\bm{R} \doteq \frac{1}{L}\bm{Y}{\bm{Y}}^*$ based on the observation matrix $\mathbf{Y}$. As mentioned earlier, while $\bm{R}$ is an unbiased estimate of $\MSIG_\nu$ under hypothesis ${\mathcal H}_{\nu}$ for $\nu = 0, 1$ when $L > M$, it becomes singular when $L < M$. Unbiasedness implies that its expected value is equal to the true covariance matrix, i.e., $\mathbb{E}(\bm{R} | {\mathcal H}_\nu) = \MSIG_\nu$. However, when the number of samples $L$ is less than the number of antennas $M$, the rank of $\bm{R}$ is at most $L$, making it non-invertible or singular. Even when $M$ is comparable to $L$, the sample covariance matrix $\bm{R}$ is not a reliable estimator for $\MSIG_\nu$, and its inverse is a poor estimator for $\MSIG_\nu^{-1}$. For instance, under the Gaussian assumption, which holds in our system model, the expected value of the inverse $\bm{R}$ is given by \cite{Baibook}:
\begin{align}
\mathbb{E} (\bm R^{-1}| {\mathcal H}_\nu)=\frac{L}{L-M-2}\MSIG_{\nu}^{-1},
\end{align}
Moreover, when $M$ is close to or larger than $L$, the expected value of the inverse $\bm{R}$ is highly biased. Specifically, for the case when $M = \frac{L}{2} + 2$, we have $\mathbb{E} (\bm{R}^{-1}| {\mathcal H}_\nu) = 2 \MSIG{\nu}^{-1}$, and this bias becomes more pronounced as $M$ approaches $L$.

Additionally, many LRTs for MASS rely on the eigenvalues of the SCM. However, when $M > L$, only the first $L$ eigenvalues of $\bm{R}$ will be non-zero, while the true covariance matrix $\MSIG_\nu$ typically has a rank of $M$ and therefore possesses $M$ non-zero eigenvalues. Moreover, as the number of antennas increases, the smallest eigenvalues of $\bm{R}$ tend to approach zero rapidly.

\subsection{Random Matrix Theory Preliminary}\label{subsec:RMT}
The key idea behind random matrix theory is to utilize limiting theorems to analyze the behavior of high-dimensional matrices. By applying these theorems, we can extend the results obtained in the asymptotic regime to matrices of lower dimensions. Remarkably, the results obtained in the high-dimensional regime provide accurate approximations even for matrices of lower dimensions. 

\textbf{Preliminary conditions}: In the context of random matrix theory, we make the following assumptions for the rest of the paper. Consider a random matrix $\mathbf{R}$ of size $p\times n$, where $p$ and $n$ tend to infinity with the same rate. We denote the ratio of the number of rows to the number of columns as $c=\frac{p}{n}$. We assume that the eigenvalues of $\mathbf{R}$ are located in a compact interval on the positive real axis, denoted by $\mathcal{T}$. Additionally, we assume that any complex function $g(z)$ used in our analysis is analytic on the positive real axis.
  	
\begin{definition}\emph{Empirical Spectral Distribution (ESD)}:
For a given $p\times p$ Hermitian matrix $\mathbf{A}$, the empirical spectral distribution (ESD) of $\mathbf{A}$ is defined as follows:
\begin{equation}
	F_{\mathbf{A}}(x) = \frac{1}{p}\sum_{i=1}^{p} \delta(x - \lambda_i),
\end{equation}

where $\lambda_1, \lambda_2, \dots, \lambda_p$ are the eigenvalues of $\mathbf{A}$, and $\delta(x)$ denotes the Dirac delta function.
\end{definition} 

\begin{definition}\emph{Limiting Spectral Distribution (LSD)}:
	Let $\{\mathbf{A}_n\}_{n=1}^{\infty}$ be a sequence of random matrices. The empirical spectral distributions of the sequence, denoted by $\{F^{\mathbf{A}_n}\}_{n=1}^{\infty}$, almost surely (a.s.) converge to a non-random distribution $F^{\mathbf{A}}$, which is called the limiting spectral distribution (LSD).
\end{definition} 

In particular, if we consider the sequence of random matrices as sample covariance matrices $\mathbf{R}_{n} = \frac{1}{n} \mathbf{X}\mathbf{X}^{*}$, where $\mathbf{X} = [X_{ij}]_{p\times n}$ and $X_{ij}$, for $i=1,\dots,p$ and $j=1,\dots,n$, are independent and identically distributed (i.i.d.) zero-mean complex Gaussian random variables with unit variance, under certain assumptions, the LSD of the sequence weakly converges to the Marcenko-Pastur distribution $F_{\text{MP}}(x)$ with the following probability density function (PDF):
\begin{align}
\label{pastureq}
\frac{d}{dx}F_{\mathcal {MP}}(x) &= \max (0,(1 - \frac{1}{c}))\delta (x) \nonumber \\&+ \frac{1}{2\pi cx}\sqrt{(x-a)^+(b-x)^+},
\end{align}
where $a= (1- \sqrt {c})^ 2$, $b = (1+ \sqrt {c}) ^2$, $(x)^{+}=\max(x,0)$, and $\delta(\cdot)$ is the Dirac delta function.

Furthermore, for the maximum and minimum eigenvalues of such SCM $\mathbf{R}_{n}$, we have the following results~\cite{mestre2008}:
\begin{align}
\label{eigenmaxmin}
&\mathop {\lim }\limits_{n \to \infty } {\lambda _{\min }}({\bm R_n}) = {(1 - \sqrt c )^2}\\
&\mathop {\lim }\limits_{n \to \infty } {\lambda _{\max }}({\bm R_n}) = {(1 + \sqrt c )^2}.
\end{align}

\begin{definition}\emph{Linear Spectral Statistic (LSS)}:
Let us assume a sequence of random matrices $\{\mathbf{A}_n\}_{n=1}^{\infty}$, where each $\mathbf{A}_n$ is a $p\times p$ matrix coming from a parametric class $\mathcal{F}=\{{\mathcal P}_\theta|\theta \in\Theta\}$. Inference about the parameter $\theta$ can be obtained from the following sufficient statistic known as the linear spectral statistic (LSS):
\begin{align}
\label{LSSeq}
\hat{\theta}_p = \int g(x) dF^{\mathbf{A}}(x) = \frac{1}{p}\sum_{i=1}^{p} g(\lambda_i)
\end{align}
where $F^{\mathbf{A}}(x)$ is the limiting spectral distribution (LSD) of the matrix $\mathbf{A}$ and $\lambda_i$, $i=1,\dots,p$, are the eigenvalues of the matrix $\mathbf{A}$.
\end{definition}

\begin{definition}\emph{Stieltjes Transform}:
If $F$ is a real, bounded, and measurable function on the real number line $\mathbb{R}$, then its Stieltjes transform is defined as
\begin{align}
m_F(z) \doteq \int_{-\infty}^{\infty} \frac{1}{\xi - z} dF(\xi), \quad z \in \mathbb{C}^+,
\end{align}
where $\mathbb{C}^+ = \{z \in \mathbb{C} : \mathrm{Im}(z) > 0\}$.
\end{definition}

Furthermore, the inverse Stieltjes transform, denoted as $z_F(m)$ for $m \in \mathbb{C}^+$, is given by
\begin{align}
z_F(m) = -\frac{1}{m} + c\int \frac{tdF^T(t)}{1 + tm}, \quad m \in \mathbb{C}^+.
\label{inversem}
\end{align}

The following theorem shows a mathematical approach to calculate LSS for high-dimensional random matrices based on RMT and, in the sequel, will be used to derive the proposed detector.
\begin{theo}\label{theo1}
Let $z\mapsto \omega=\omega_p(z)$ be a mapping and its inverse $\omega\mapsto z=z_p(\omega)$ defined by
\begin{align}
\label{phieq}
z_p(\omega)=\omega\left(1-\frac{1}{n}\sum_{i=1}^p\frac{\lambda_i}{\lambda_i-\omega}\right),
\end{align}
where $\lambda_i$ are the eigenvalues of the matrix $\bm A$ and $n$ is the dimension of the matrix. Let $\phi_p(\omega)= g(z_p(\omega))$ and $\mathcal{C}_{\omega}^+=\omega_p(\mathbb{C}^+)$, where $\mathbb{C}^+$ is the upper complex half-plane. 
Suppose the preliminary conditions hold, and let $n$ and $p$ be sufficiently large. Then, the LSS in (\ref{LSSeq}) almost surely (a.s.) converges to the following expression:
\begin{align}
\label{main_int}
\hat{\eta}=(1-\frac{n}{p})\frac{1}{2\pi j}\oint_{\mathcal{C}_{\omega}^+}\frac{g(z)}{z}dz
+\frac{n}{p}.\frac{1}{2\pi j}\oint_{\mathcal{C}_{\omega}^+}\frac{\phi_p(\omega)}{\omega}\psi_p(\omega)d\omega,
\end{align}
where
 \begin{align}
\label{zchange}
&\psi_p(\omega)=\frac{\mathrm{d}z_p(\omega)}{\mathrm{d}\omega}=1-\frac{1}{n}\sum_{i=1}^p\frac{\lambda_i^2}{(\lambda_i-\omega)^2},
\end{align}
and $\mathcal{C}_{\omega}^+$ is a positive contour includes all eigenvalues.
\end{theo}
\begin{proof}
see~\cite{mestre2017}
\end{proof}

The following theorem, which is equivalent to the CLT for random variables, is important and useful in RMT. Here, we present one variation of the theorem for the LSS, which will be used in Section \ref{sec:pfa} to derive the distribution of the proposed detector and evaluate its performance.

Suppose the matrix $\bm X$ is constructed as described in (\ref{pastureq}), and let $\mathcal{U}$ be an open set in the complex plane that includes the support of the Marcenko-Pastur distribution. Let $\mathcal{A}$ be the collection of all analytic functions $g: \mathcal{U} \rightarrow \mathbb{C}$. Then, we define the empirical spectral process as follows:

\begin{align}
D_n(g) \doteq p\int_{-\infty }^{ + \infty } {h(x)d[{F_n}(x) - {F_{\mathcal {MP}}}(x)]},g\in\mathcal{A}, 
\end{align}

where $ F_n$ is the ESD of the matrix ${\bm R}_n$, we will have the following result.

\begin{theo}
\label{CLT}
 Central Limit Theorem (CLT) for LSS: If $\EX( x_{11})=0,\quad\EX|x_{11}|^2=1,\quad\EX |x_{11}|^4<\infty$, then for $g_1,\dots,g_k\in\mathcal{A}$ random vector $(D_n(g_1),\dots,D_n(g_k))$, converges to a Gaussian random vector indexed by $g_i$ with mean and covariance
\small
\begin{align}
\label{MEAN}
\mu ({g_j})&=(\kappa-1)\big(\frac{{{g_j}(a(c)) + {g_j}(b(c))}}{4}  \\
&-\frac{1}{{2\pi }}\int_{a(c)}^{b(c)} {\frac{{{g_j}(x)}}{{\sqrt {4y - {{(x - 1 - c)}^2}} }}dx}\big),\nonumber\\
Cov({g_j},{g_i }) &=  - \frac{\kappa}{{4{\pi ^2}}}\oint {\oint {\frac{{{g_j}({z_1}){g_i }({z_2})}}{{{{(\underline{m}({z_1}) -\underline{m}({z_2}))}^2}}}} } d\underline{m}({z_1})d\underline{m}({z_2}),
\label{COV}
\end{align}
\normalsize
for $j,i\in\{1,\dots,k\}$. For real case $\kappa=2$
 and complex case $\kappa=1$ and $\underline{m}(z)$ is the Stieltjes transform of $\underline{F}\doteq(1-c)\bm{1}_{[0,\infty)}+cF_{\mathcal {MP}}$ and contours in ~\ref{COV} are non-overlapping and both contain the support of $F_{\mathcal {MP}}$. Moreover, if $\{x_{ij}\}_{i=1,..,n}^{j=1,..,p}$ to be complex then the mean is the same of real case and the variance is the half variance in~(\ref{COV}).
\end{theo}

\begin{proof}
See \cite{Baibook}
\end{proof}
\section{Proposed Detectors for Large Array Antennas}\label{sec:detectors}
In this section, we propose multiple antenna detectors for large array antennas based on RMT. These detectors are designed for the scenario where the noise covariance matrix is white. Our analysis is based on the LSS introduced in the previous section. The LSS represents a linear combination of functions of the eigenvalues of a high-dimensional matrix.

For MASS problem, we need to decide between two hypotheses, ${\mathcal H}_0$ and ${\mathcal H}_1$, and the decision statistics based on LSS have the general form given by equation (\ref{LSSeq}), where $\hat{\theta}_M$ is the LSS estimate and $g(\lambda_i)$ represents the function evaluated at the $i$-th eigenvalue of the SCM $\bm R$. 

It is important to highlight that many existing eigenvalue-based detectors for MASS can be expressed using the LSS function. In this section, we extend the analysis to incorporate the large array antenna scenario by exploiting RMT and Theorem 2. The key concept is to demonstrate that as the number of antennas ($M$) and samples ($L$) increase, with the ratio $\frac{M}{L}$ converging to a constant $c \in (0,+\infty)$, the sequence of random LSS functions in (\ref{LSSeq}) converges almost surely to a limiting statistic $\hat{\eta}$. This limiting statistic, $\hat{\eta}$, can serve as a reliable approximation for the decision statistics in large array antenna spectrum sensing, particularly when the number of antennas is comparable to or greater than the number of samples.

\subsection{Linear Function: High-Dimensional GLR Test}
In this subsection, we focus on the linear function case, where the function $g(z)$ takes the form $g(z) = z$. We derive a high-dimensional version of this detector using the LSS framework. The LSS function in this case can be expressed as:
\begin{align}
\hat{\theta}_M = \frac{1}{M}\sum_{i=1}^M \lambda_i,
\label{LSS_linear}
\end{align}
where $\lambda_i$ represents the $i$-th eigenvalue of the SCM $\bm R$.

The above LSS function corresponds to the comparison of the normalized summation of eigenvalues. Several existing MASS techniques can be formulated in this framework. For example, the GLRT proposed in \cite{taherpour2010}, can be expressed as:
\begin{align}
T_{GLR} = \frac{\lambda_{\mathrm{max}}}{\sum_{i=1}^M \lambda_i},
\label{GLR_test}
\end{align}
where $\lambda_{\mathrm{max}}$ is the maximum eigenvalue of the SCM. As the number of antennas $M$ becomes larger, according to Equation (\ref{eigenmaxmin}), we observe that $\lambda_{\mathrm{max}}$ approaches a constant value. Hence, for large values of $M$, we can approximate $\lambda_{\mathrm{max}}$ by $(1-\sqrt{c})^2$. By taking the inverse of the test statistic, we obtain the summation of eigenvalues as the test statistic. It is worth mentioning that several other detectors can also be formulated in this framework \cite{larsonICASSP,ramirez2011detection}.
 
To obtain the large array antenna version of the linear function detectors, we consider $g(z) = z$ in (\ref{phieq}) which by substituting, we have:
\begin{align}
\phi_M(\omega) = \omega - \frac{\omega}{L}\sum_{i=1}^M\frac{\lambda_i}{\lambda_i-\omega_M(\omega)},
\end{align}
where $\omega_M(\omega)$ is the inverse of the Stieltjes transform $m_{\bm R}(\omega)$.

Upon conducting mathematical computations, we can establish the following theorem:
\begin{theo}
The high-dimensional decision statistic for a linear function of LSS has the following form:
\begin{align}
T_{HDL} = \frac{1}{M}\TRAC(\bm{R}) = \frac{1}{M}\sum_{i=1}^{M}{\lambda_i},
\label{HDL_det}
\end{align}
where $\TRAC(\bm{R})$ represents the trace of the SCM $\bm R$.
\end{theo}

\begin{proof}
See Appendix \ref{app_1}.
\end{proof}

As we can observe, this detector exhibits a similar form to the GLRT used in conventional MASS when $L>M$. Therefore, we can conclude that the (normalized) summation of the eigenvalues of the sample covariance matrix can serve as the decision statistic for MASS, regardless of the number of antennas and samples.

\subsection{Square Function: High-Dimensional Frobenius Norm Test}
For the square function $g(z) = z^2$, the corresponding LSS function takes the form:
\begin{align}
\hat{\theta}_M = \frac{1}{M}\sum_{i=1}^M \lambda_i^2,
\label{FNT_LSS}
\end{align}
which represents the summation of the squared eigenvalues of the SCM, normalized accordingly.

In the context of MASS, various detectors have been proposed based on the higher-order moments of the eigenvalues of the SCM. One popular choice is the Frobenius norm test, which uses the squared Frobenius norm of the SCM as the test statistic and rejects the null hypothesis for sufficiently large values of:
\begin{align}
T_{FN} = \frac{1}{M}\|\bm R\|_{F}^{2} = \frac{1}{M}\TRAC(\bm{R}^2) = \frac{1}{M}\sum_{i=1}^{M}{\lambda _{i}^2},
\label{FNT}
\end{align}
which has the same form as (\ref{FNT_LSS}).

To obtain the high-dimensional version of detectors based on the square function $g(z) = z^2$, we substitute this function into the expression in Equation (\ref{phieq}). After some mathematical calculations, we obtain the following result:

\begin{align}
\phi_M(z) = z^2 = \left(\omega - \frac{\omega}{L}\sum_{i=1}^M\frac{\lambda_i}{\lambda_i - \omega_M(z)}\right)^2.
\end{align}

This expression represents the high-dimensional version of the LSS function for the square function. By using this result, we can derive the corresponding detectors for a large array antenna scenario.

\begin{theo} The high-dimensional test for Frobenius norm test has the following form in large array antenna
 \begin{align}
\label{HDFNT_det}
 T_{HDS}&=\frac 1 M\TRAC(\bm{R}^2)+\frac{1}{LM}(\TRAC(\bm{R}))^2\nonumber\\
&=\frac 1 M\sum_{i=1}^{M}{\lambda _i^2}+\frac{1}{LM}(\sum_{i=1}^{M}{\lambda _i})^2\\
&=\frac 1 M (1+ \frac 1 L)\sum_{i=1}^{M}{\lambda _i^2}+\frac{2}{LM}\sum_{k>i}^{M}{\lambda _i \lambda_k} \nonumber,
 \end{align}
 \end{theo}
\begin{proof}
See Appendix~ \ref {app_2}
\end{proof}

In the case of the square function $g(z) = z^2$, the decision statistic for a large array antenna scenario is different from the conventional multiple antenna case. This decision statistic captures the squared eigenvalues of the sample covariance matrix. It is worth noting that the presence of cross-terms in the decision statistic arises from the second term in the expression, which involves the interaction between different eigenvalues. As the number of samples $L$ increases while keeping the number of antennas $M$ fixed, this term becomes smaller and asymptotically tends to the same form as the conventional Frobenius norm test.

\subsection{Quadratic Function: High-Dimensional Complex Rao Test}
In this section, we consider a quadratic function of the form $g(z) = (z-1)^2$, which combines the characteristics of the two previous cases.
 \begin{align}
\hat{\theta}_M=\frac{1}{M}\sum_{i=1}^M (\lambda_i-1)^2.
\label{QLSS}
\end{align}
 
In fact, the proposed MASS based on Complex Rao test in \cite{taherpourCL} has the following from:
\begin{align}
T_{\rm Rao}=\TRAC\left[(\bm{R}-\IDM)^2\right]
\end{align} 
From above, we can rewrite the detector $T_{\rm Rao}$ as
 \begin{align}
 T_{\rm Rao}=\TRAC[(\bm{R}-\IDM)^2]=M\frac{1}{M}\sum_{i=1}^M{(\lambda_i-1)^2}
 \end{align}
 So for this detector, it can be seen that $g(z)=(z-1)^2$. Thus, from (\ref{phieq}), we will have
 \begin{align}
  \phi_M(z)=(z-1)^2=&\big(\omega-1-\frac{\omega}{L}\sum_{i=1}^M\frac{\lambda_i}{\lambda_i-\omega_M(z)}\big)^2\nonumber\\=&\omega^2\big(\frac{\omega-1}{\omega}-\frac{1}{L}\sum_{i=1}^M\frac{\lambda_i}{\lambda_i-\omega_M(z)}\big)^2,
 \end{align}
which from (\ref{main_int}), we finally have the following result.

\begin{theo} The high-dimensional complex Rao test has the following form
\begin{align}
 \label{HD-RT1}
T_{HDQ}&= \frac 1 M\TRAC(\bm{R}^2)-\frac 2 M\TRAC(\bm{R})+\frac{1}{LM}(\TRAC(\bm{R}))^2 \nonumber\\
&=\frac 1 M\sum_{i=1}^{M}{\lambda _i^2}-\frac 2 M\sum_{i=1}^{M}{\lambda _i}+\frac{1}{LM}(\sum_{i=1}^{M} {\lambda _i})^2 \\
&=\frac 1 M \sum_{i=1}^{M}{(\lambda _i -1)^2}+\frac{2}{M(L+1)}\sum_{k>i}^{M}{\lambda _i \lambda_k} \nonumber,
 \end{align}
 \end{theo}
\begin{proof}
See Appendix~ \ref {app_3}
\end{proof}

The proposed detector for the large array antenna case differs from conventional multiple antenna detectors in that it incorporates both individual eigenvalues and cross-terms that capture the interactions between different eigenvalues. This quadratic function of the eigenvalues of the sample covariance matrix distinguishes it from existing detectors in conventional multiple antenna scenarios.

As the number of samples $L$ increases while keeping the number of antennas $M$ fixed, the impact of cross-terms in the decision statistic diminishes. This is due to the fact that cross-terms are determined by the interaction between different eigenvalues, and as the number of samples increases, the contributions from individual eigenvalues become more dominant. Consequently, in the limit of a large number of samples, the decision statistic converges to the form observed in conventional multiple antenna scenarios, where only individual eigenvalues are significant.

This behavior highlights the influence of large array antennas on the decision statistic. With an increased number of antennas $M$, more eigenvalues are involved in the summation, resulting in the presence of cross-terms. However, as the number of samples grows, the impact of cross-terms diminishes, and the decision statistic approaches the conventional case.

\begin{rmk}
The computational complexities of the derived detectors for large array antennas are indeed in the same order or lower than their low-dimensional counterparts. The complexity of computing the sample covariance matrix (SCM) is $\mathcal{O}(M^2L)$, where $M$ is the number of antennas and $L$ is the number of samples. The complexity of computing the squared SCM is $\mathcal{O}(M^{2.37})$. In comparison, some conventional MASS detectors require obtaining the eigenvalues of the SCM using methods such as singular value decomposition (SVD) or similar techniques, which have a computational complexity of $\mathcal{O}(M^3)$ or lower. However, for the proposed detectors, there is no need to explicitly compute the eigenvalues. Instead, the decision statistics can be calculated directly using the eigenvalues of the SCM and its square. This eliminates the need for eigenvalue computation, resulting in lower computational complexity. Overall, the proposed detectors for large array antennas offer computational advantages over certain conventional detectors, as they can be computed using the SCM and its square without requiring explicit eigenvalue calculations.
\end{rmk}
\section{Distributions Under Noise Hypothesis: False Alarm Probability}\label{sec:pfa}
To obtain the decision threshold using the Neyman-Pierson (NP) method, we need to compute the distribution of test statistics under hypothesis ${\mathcal{H}}_0$. So, in this section, we investigate the distribution of the proposed high-dimensional detectors under hypothesis ${\mathcal{H}}_0$.

One of the major differences of the results for the analysis of large array antenna detectors by using RMT from the conventional ones based on LRT is their distribution. In LRT obtained detectors, when the decision statistics are derived from methods such as GLR, Wald test, and the Rao tests, asymptotically, the decision statistics will be distributed as a central chi-square distribution. In fact, if we assume ${\cal L}_L$ shows the likelihood ratio function, and under hypothesis ${\cal H}_0$, for fixed number of antennas $M$, when $L\rightarrow \infty$ then $-2\ln({\cal L}_L)\rightarrow \chi_f^2$ to chi-square distribution with degrees of freedom of $f=\frac{1}{2}M(M+1)$. 

However, in the high-dimensional case when $M$ is large, the distribution of the decision statistics is Gaussian, and it is described by its mean and variance, and these statistical parameters should be computed by using the complex integrals and complicated expressions which just recently been studied and matured in the area of the RMT. So in this section, using the theorem~\ref{CLT} and (\ref {MEAN}) and (\ref{COV}), we obtain the corresponding mean and variance of three proposed detectors. It is notable that, for the results of this section, we assume that noise variance is one.

\begin{theo}
Asymptotic distribution of decision statistic of $T_{HDL}$ under hypothesis $\mathcal{H}_0$ is Gaussian with the following mean and variance
\begin{align}
T_{HDL}| {\mathcal H}_0\stackrel{\text{ap.}}{\thicksim}\mathcal{N}\big(M~,~c\big),
\label{dishdglrt}
\end {align}
where $ c=\frac{M}{L}$ is the ratio of number of antennas to samples.
\end{theo}
\begin{proof}
See  Appendix~\ref{App-B}.
\end{proof}
\begin{theo}
Asymptotic distribution of statistic
$T_{HDS}$ is obtained as follows
\begin{align}
T_{HDS}| {\mathcal H}_0\stackrel{\text{ap.}}{\thicksim}\mathcal{N}\big(M(1+c)~,~4c^3+10c^2+4c\big)
\label{dishdfn}
\end {align}
\end{theo}
\begin{proof}
See  Appendix~\ref{App-B}.
\end{proof}
\begin{theo}
The decision statistic of $T_{HDQ}$ under hypothesis ${\cal H}_0$, asymptotically has Gaussian distribution as
\begin{align}
T_{HDQ}| {\mathcal H}_0\stackrel{\text{ap.}}{\thicksim}
\mathcal{N}\big(Mc~,~2c^2(1+2c)\big),
\label{dishdrao}
\end {align}
\end{theo}
\begin{proof}
See  Appendix~\ref{App-B}.
\end{proof}

From the above expressions, we can see that while the number of samples $L$ impacts the performance only via $c$, the number of antenna $M$ has a dual effect by changing $c$ and the mean of the distributions. Hence intuitively, we can infer that changing $M$ impacts detectors' performance more than $L$, and below, we discuss this intuition in more detail.

As a reminder, the receiver operating characteristic (ROC) curve which illustrates the relation between detection probability $P_d$ and false alarm probability $P_{\rm fa}$, is a continuous concave function. Furthermore, if $\tau$ be the decision threshold, then $\frac{\partial P_{d} (\tau)}{\partial P_{\rm fa}(\tau)}=\tau$ or the slope of tangent to ROC curve at specific false alarm probability\cite{kaybook}. For the decision threshold in NP criterion and for specific $P_{\rm fa}=p$ for Gaussian distributions, we have $\tau=\sigma Q^{-1} (p)+\mu$, where $\mu$ and $\sigma$ denote the mean and standard deviation of Gaussian distributions and $Q(.)$ is $Q$-function for normal distribution. 
\begin{table}[!t]
    \caption{Detectors Parameters Partial Derivatives}
\centering 
    \label{tab:table1}
    \begin{tabular}{|c|c |c| c|}
    \hline\hline 
      Detector &$\frac{\partial \sigma}{\partial c}$& $\frac{\partial \mu}{\partial L}$ & $\frac{\partial \mu}{\partial M}$
    \\ 
      \hline
      $T_{HDL}$ & $\frac{1}{2\sqrt{c}}$ & 0 & $1$\\ 
      $T_{HDS}$ & $\frac{6c^2+10c+2}{\sqrt{4c^3+10c^2+4c}}$ & $-c^2$&$1+2c$\\ 
      $T_{HDQ}$ & $\frac{3c+2}{\sqrt{2(1+c)}}$ & $-c^2$& $2c$ \\ 
      \hline
    \end{tabular}
\end{table}
Now for detection probability variation at the given point over the ROC curve by changing parameters $L$ and $M$, we have
\begin{align}
&\frac{\partial(\frac{\partial P_{d}}{\partial P_{\rm fa}})} {\partial L}=\frac{\partial \tau}{\partial L}=\frac{\partial \sigma}{\partial L} Q^{-1}(p)+\frac{\partial \mu}{\partial L}\nonumber\\&=\frac{\partial \sigma}{\partial c}.\frac{\partial c}{\partial L}Q^{-1}(p)+\frac{\partial \mu}{\partial L}=-\frac{1}{L}c\frac{\partial \sigma}{\partial c}Q^{-1}(p)+\frac{\partial \mu}{\partial L}
\end{align}
and
\begin{align}
&\frac{\partial(\frac{\partial P_{d}}{\partial P_{\rm fa}})} {\partial M}=\frac{\partial \tau}{\partial M}=\frac{\partial \sigma}{\partial M} Q^{-1}(p)+\frac{\partial \mu}{\partial M}\nonumber\\&=\frac{\partial \sigma}{\partial c}.\frac{\partial c}{\partial M}Q^{-1}(p)+\frac{\partial \mu}{\partial M}
=\frac{1}{L}\frac{\partial \sigma}{\partial c}Q^{-1}(p)+\frac{\partial \mu}{\partial M}
\end{align}
Note that from Table \ref{tab:table1}, we can witness that for all of the detectors, $\frac{\partial(\frac{\partial P_{d}}{\partial P_{\rm fa}})} {\partial L}<0$ and $\frac{\partial(\frac{\partial P_{d}}{\partial P_{\rm fa}})} {\partial M}>0$. This implies that for fixed $M$, by increasing $L$ the change rate in the slop of tangent line to ROC curve (threshold) decreases and the curves converge together and hence the performance improvement becomes smaller for higher number of samples, $L$. In contrary for fixed $L$, in the region of very low false alarm probability $p\ll1$, by increasing $M$, the slope change rate is increasing and hence the performance improves. However as $p \rightarrow 1$ the slope change rate is getting smaller and negligible and since the ROC curve is continuous and concave, they converge to the point  $(P_{\rm fa}=1 , P_d=1)$.   
\begin{rmk}
Based on the analysis, the number of antennas $M$ has a more significant impact on the performance of MASS compared to the number of samples $L$, especially when aiming for low false alarm probabilities. Increasing $M$ allows for a higher detection threshold, which can be achieved by increasing the mean of the Gaussian distributions in the detectors. This, in turn, reduces the false alarm probability. On the other hand, increasing $L$ has a smaller effect on the performance. However, it is essential to consider practical limitations when determining the number of antennas. Factors like cost, physical space, power consumption, and system complexity may impose constraints on the deployment of larger array antennas. Therefore, there is a trade-off between the potential performance improvement gained by increasing $M$ and the practical limitations of the system.

In summary, if the goal is to achieve accurate spectrum sensing with low false alarm probabilities, increasing the number of antennas should be considered to improve performance, as long as the practical constraints of the system allow for the deployment of larger array antennas.
\end{rmk}
 \begin{figure}[!t]
                \centering
                \includegraphics[width=.9\columnwidth ,height=2 in]{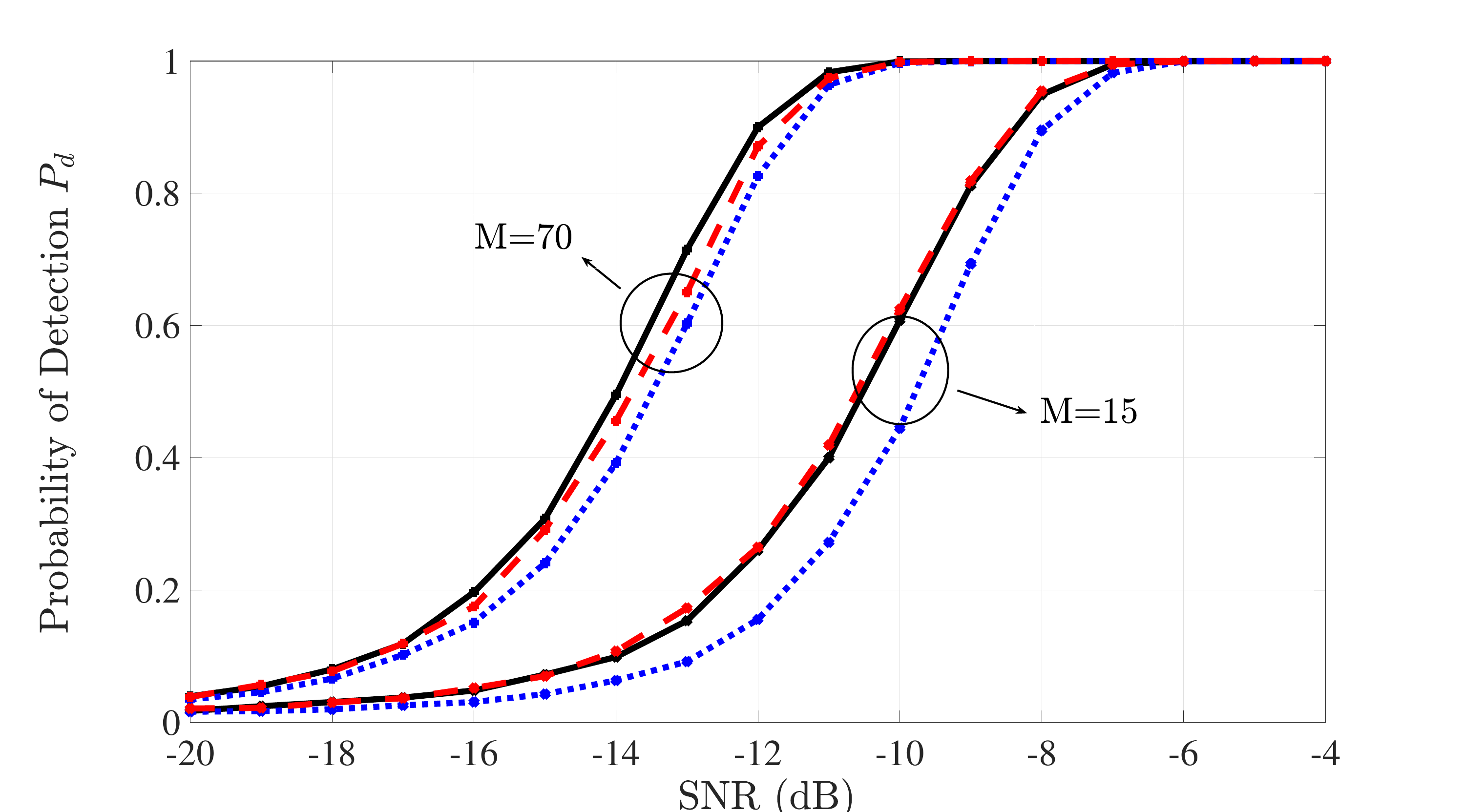}
\DeclareGraphicsExtensions.
\caption{The detection probability versus SNR for $P_{\rm fa}=0.01$, $L=30$ and for two case of $M=15$ and $M=70$. Solid line ($T_{HDL}$), dashed line ($T_{HDQ}$) and dotted line ($T_{HDS}$).}
\label{snrantenna}
\end{figure}
\section{Simulation Results}\label{sec:simulation}
In this section, we assess the performance of the proposed detectors and compare the obtained analytical results with Monte-Carlo simulations. We define the signal-to-noise ratio (SNR) for the simulation results in this section as follows:
\begin{align}
\overline{\gamma}=\frac{\mathcal{E}_s\mathrm{tr}(\bm{\Sigma}_H)}{{\rm tr}(\bm{\Sigma}_N)}
 \label{SNR}
 \end{align}
We utilize conventional uncorrelated and correlated MIMO channel models for simulations since there is no specific channel model for massive MIMO systems. Fig. \ref{snrantenna} illustrates the performance of the proposed detectors $T_{HDL}$, $T_{HDS}$, and $T_{HDQ}$ as a function of SNR for two scenarios: $M=15$ and $M=70$, with $P_{\rm fa}=0.01$ and $L=30$. For $P_{\rm fa}=0.01$ and $\overline{\gamma}=-10$ dB, all detectors achieve a detection probability exceeding 99\% when $M=70$. Similarly, in the case of $M=15$, both $T_{HDL}$ and $T_{HDQ}$ exhibit similar performance and outperform $T_{HDS}$. However, as the number of antennas increases from $M=15$ to $M=70$, corresponding to higher values of $c$, the performance gap between $T_{HDS}$ and the other two detectors diminishes. Furthermore, in this scenario, $T_{HDL}$ demonstrates slightly superior performance compared to $T_{HDQ}$.
 
 \begin{figure}[!t]
\centering
 \subfigure[t][  ]{
                \centering
                \includegraphics[width=2.5in ,height=2 in]{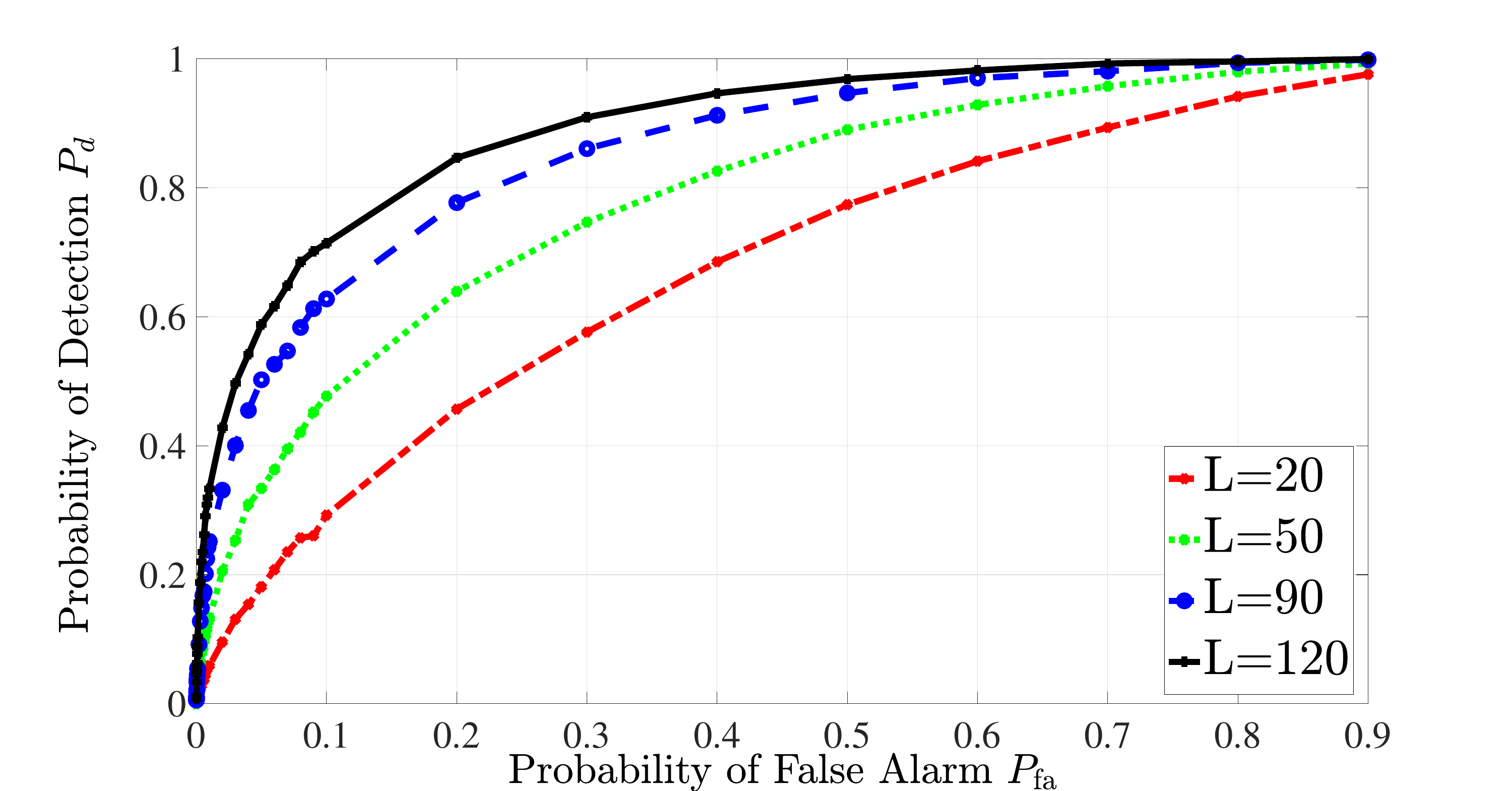}
                \label{GLRsample}
                }
               \subfigure[t][ ]{
                \centering
                \includegraphics[width=2.5 in ,height=2 in]{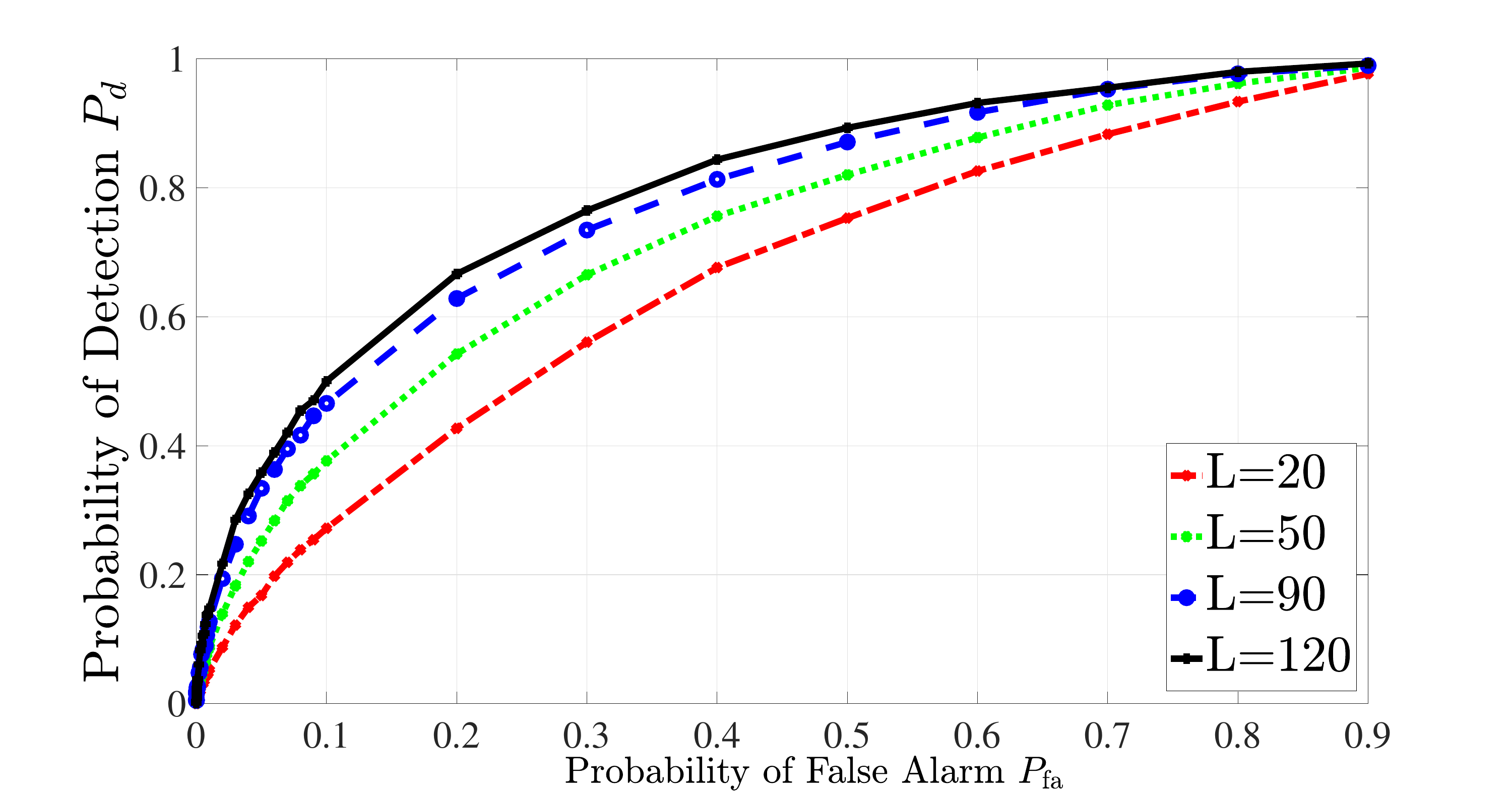}
                \label{FNsample}
                }      
        \subfigure[t][ ]{
                \centering
                \includegraphics[width=2.5 in,height=2 in]{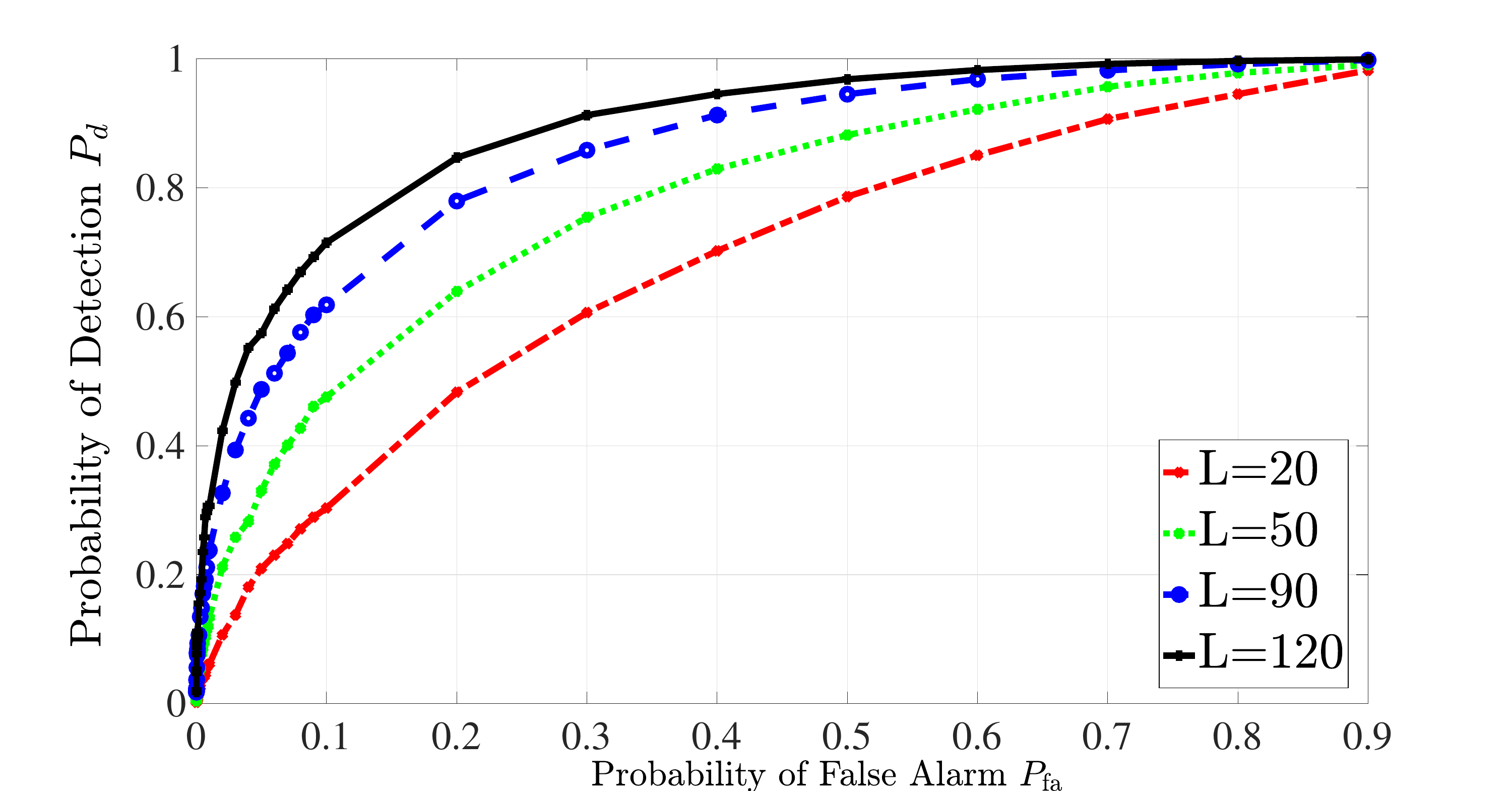}
                \label{Raosample}
                }
\DeclareGraphicsExtensions.
\caption{ROC of proposed detectors for the case of correlated antennas for different values of number of sample $L$, $M=30$ and $\overline{\gamma}=-15dB$, a) $T_{HDL}$, b) $T_{HDS}$ and c)$T_{HDQ}$}
\label{detectorsample}
\end{figure}
 
Figure \ref{detectorsample} examines the effect of increasing the number of samples on three proposed detectors for correlated antennas, with a fixed number of antennas $M=30$ and SNR $\overline{\gamma}=-15$ dB. It can be observed that as the number of samples increases and the ratio $c$ decreases, the detectors' performance improves. However, this improvement becomes less significant as the number of samples becomes larger. This trend is particularly evident in Figure \ref{FNsample}, where increasing the number of samples from $L=90$ to $L=120$ results in only a 2\% improvement. Figures \ref{snrantenna} and \ref{detectorsample} further support our previous argument regarding the greater impact of increasing the number of antennas compared to the number of samples for achieving more accurate spectrum sensing and higher detection probabilities.
 
 \begin{figure}[!t]
\centering{
 \subfigure[t][ 
    $L=50$ and  $M=45$
 ]{
                \centering
                \includegraphics[width=.8\columnwidth ,height=2 in]{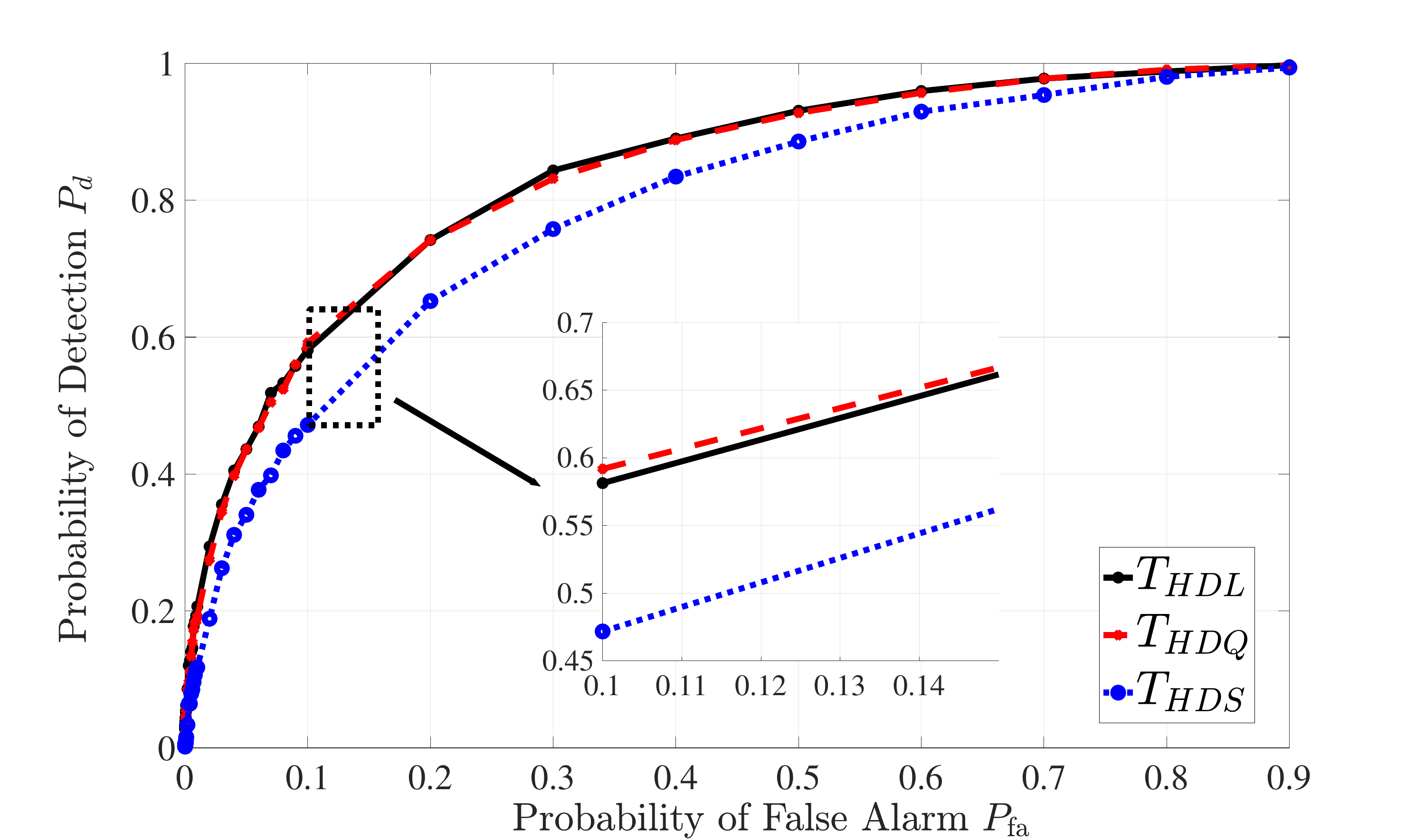}
                \label{lowcantenna}
                }
               \subfigure[t][ 
    $L=50$ and  $M=70$
 ]{
                \centering
                \includegraphics[width=.8\columnwidth ,height=2 in]{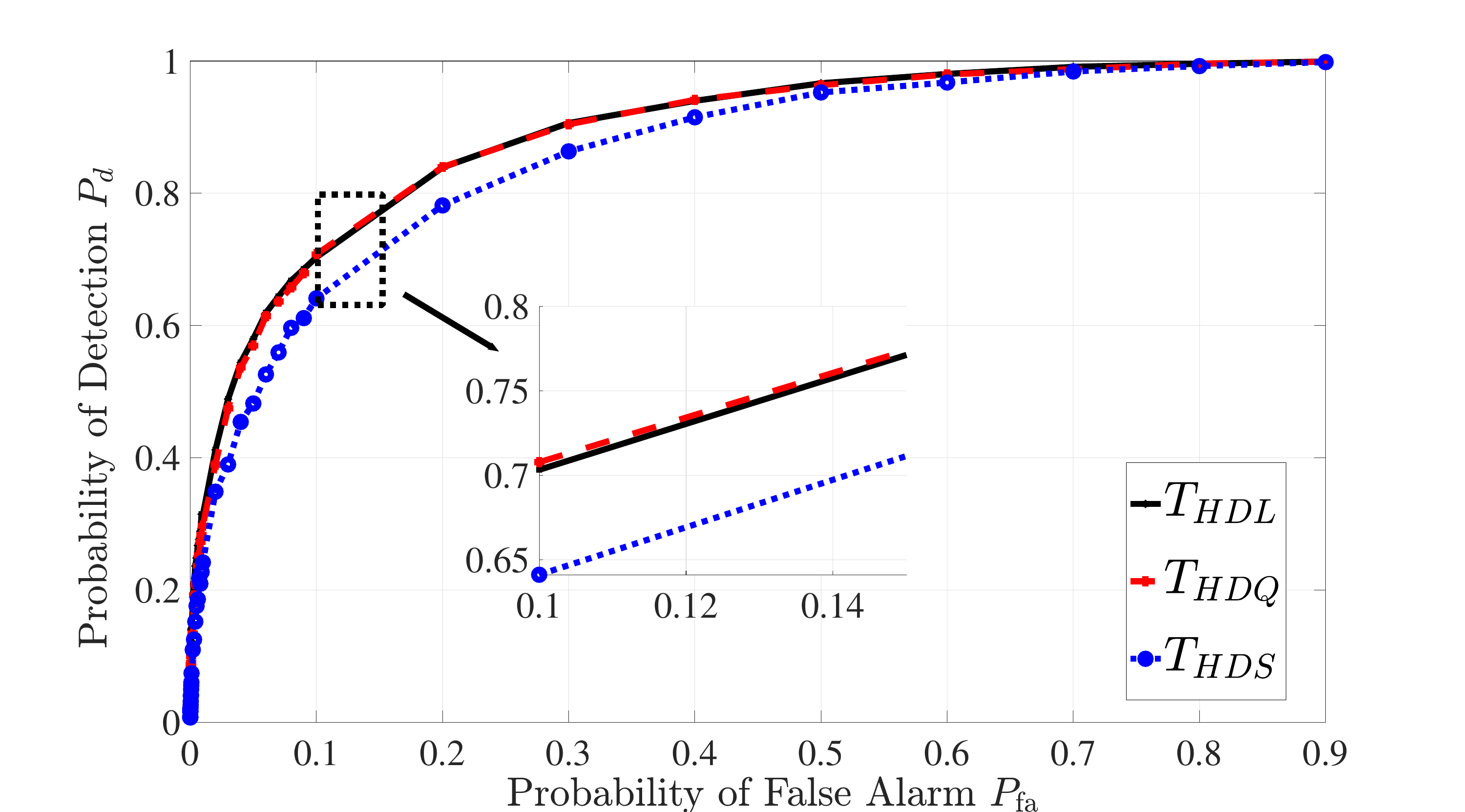}
                \label{highcantenna}
                }     
                }
\DeclareGraphicsExtensions.
\caption{ROC of proposed detectors for a) c=0.9 b) c=1.4}
\label{detectorscomp}
\end{figure}
 
 \begin{figure}[!t]
\centering
\includegraphics[width=0.8\columnwidth ,height=2 in]{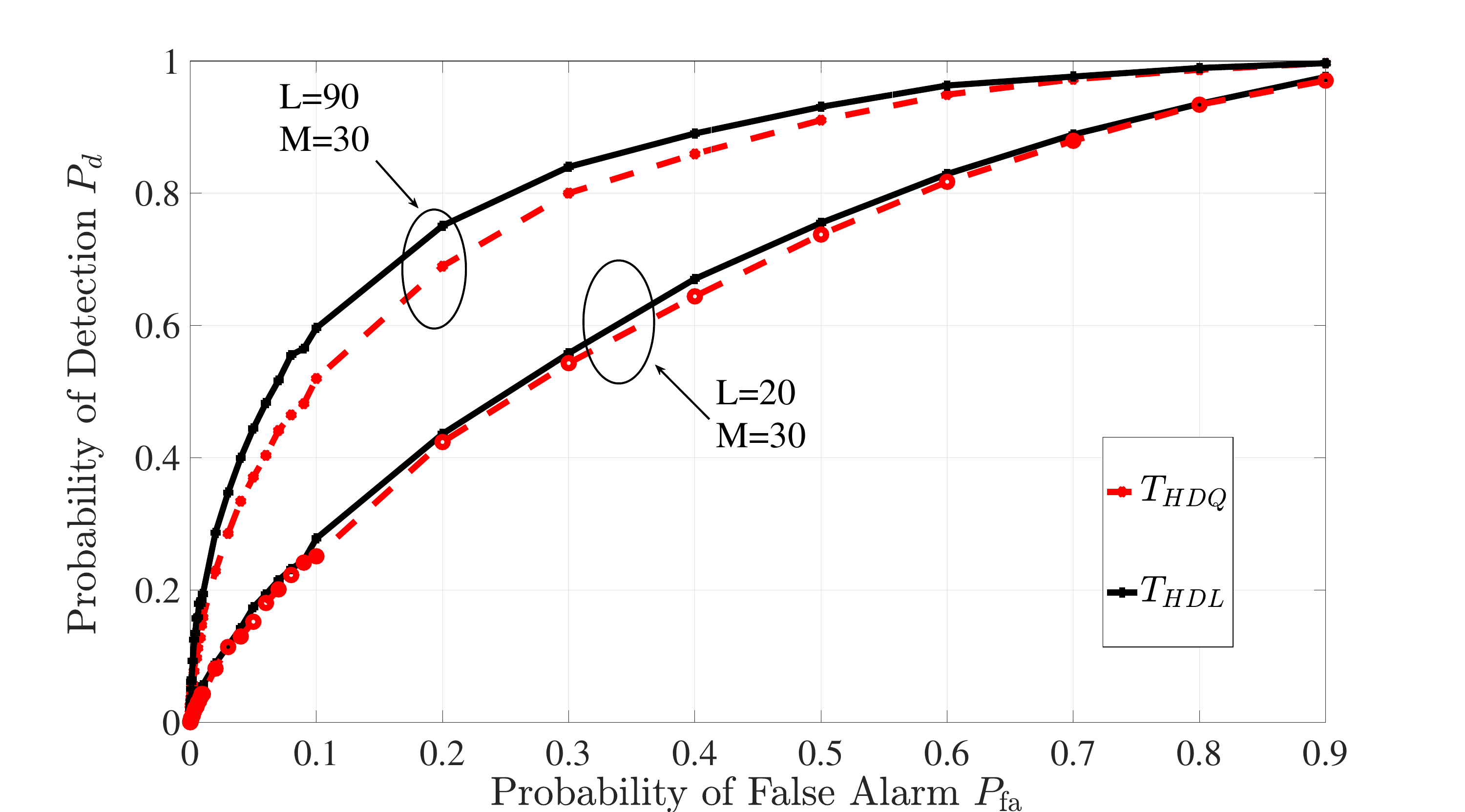}
\caption{The comparison of $T_{HDL}$ and $T_{HDQ}$ for the uncalibrated and correlated antennas }
\label{uncalib}
\end{figure}

Figure \ref{detectorscomp} illustrates the Receiver Operating Characteristic (ROC) curves for two cases: $L=50$, $M=45$ and $L=50$, $M=70$, at a low SNR of $\overline{\gamma}=-15$ dB. It is evident that the detectors $T_{HDL}$ and $T_{HDQ}$ exhibit similar performance, outperforming $T_{HDS}$. Due to its simpler form, $T_{HDL}$ is a preferable choice in calibrated antenna cases. Furthermore, as the ratio of antennas to samples ($c$) increases from $c=0.9$ to $c=1.4$, the performance gap between $T_{HDS}$ and $T_{HDL}$ diminishes.

In addition, Figure \ref{uncalib} compares the performance of detectors $T_{HDQ}$ and $T_{HDL}$ in the scenario of uncalibrated antennas. To simulate this scenario, the noise covariance matrix is assumed to be diagonal with randomly generated elements. It is observed that the performance of the detectors degrades compared to the calibrated antenna cases. For instance, with the same number of antennas, samples, and SNR, the detection probability of $T_{HDQ}$ and $T_{HDL}$ reduces by approximately 6\% and 4\%, respectively. Interestingly, in the absence of calibration, the simpler detector $T_{HDL}$ outperforms $T_{HDQ}$. Moreover, it is noteworthy that with a fixed number of antennas and an increase in the number of samples from $L=20$ to $L=120$, the performance gap between the two detectors in the uncalibrated case widens.
\begin{figure}[!t]
                \centering
                \includegraphics[width=.8\columnwidth ,height=2 in]{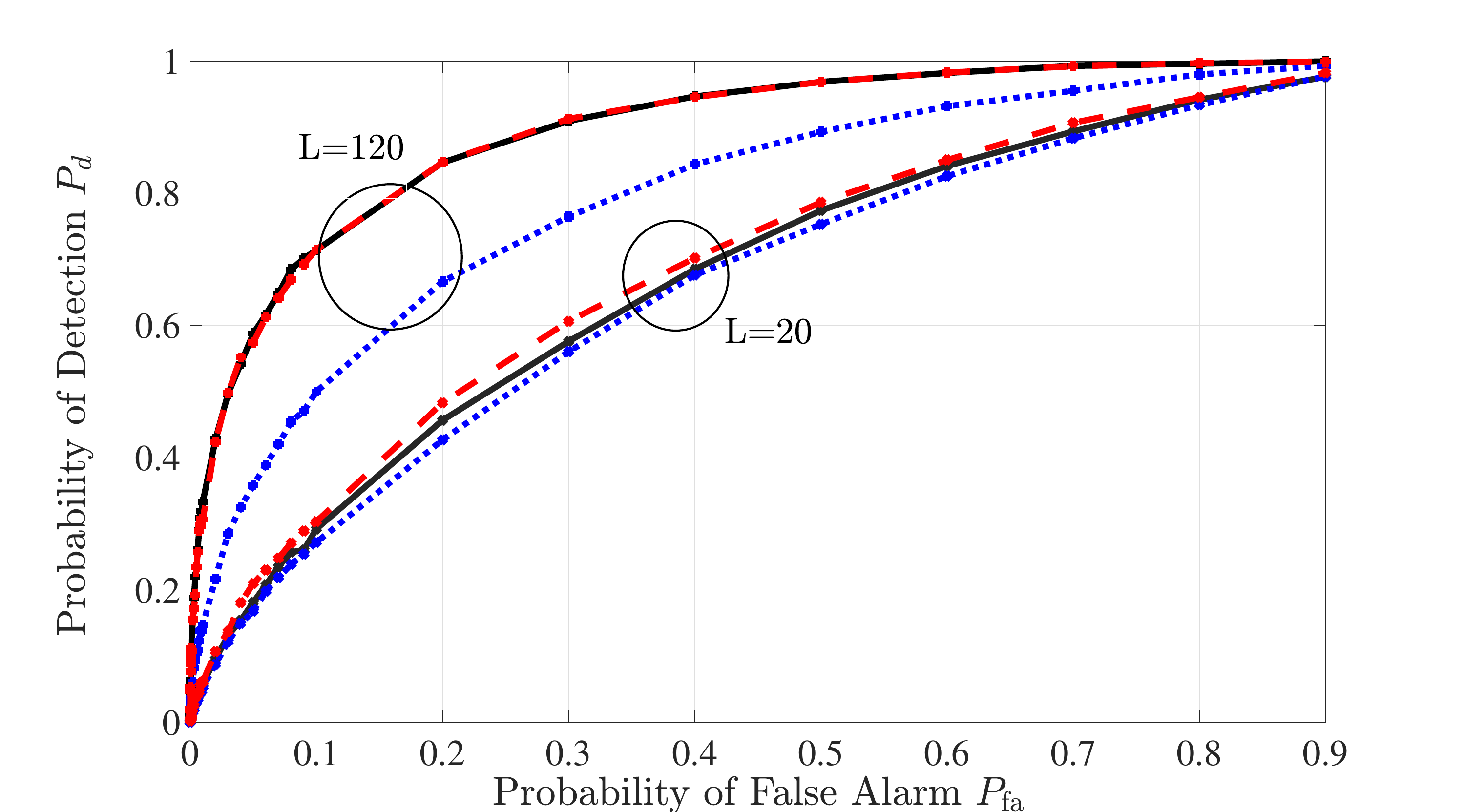}
                \label{thglrd31}
\DeclareGraphicsExtensions.
\caption{ROC of the proposed detectors for different number of samples and $P_{\rm fa}=0.01$ and  $\overline{\gamma}=-10dB $}
\label{ROCsample}
\end{figure}
 \begin{figure}[!t]
\centering
                \includegraphics[width=.9\columnwidth ,height=2 in]{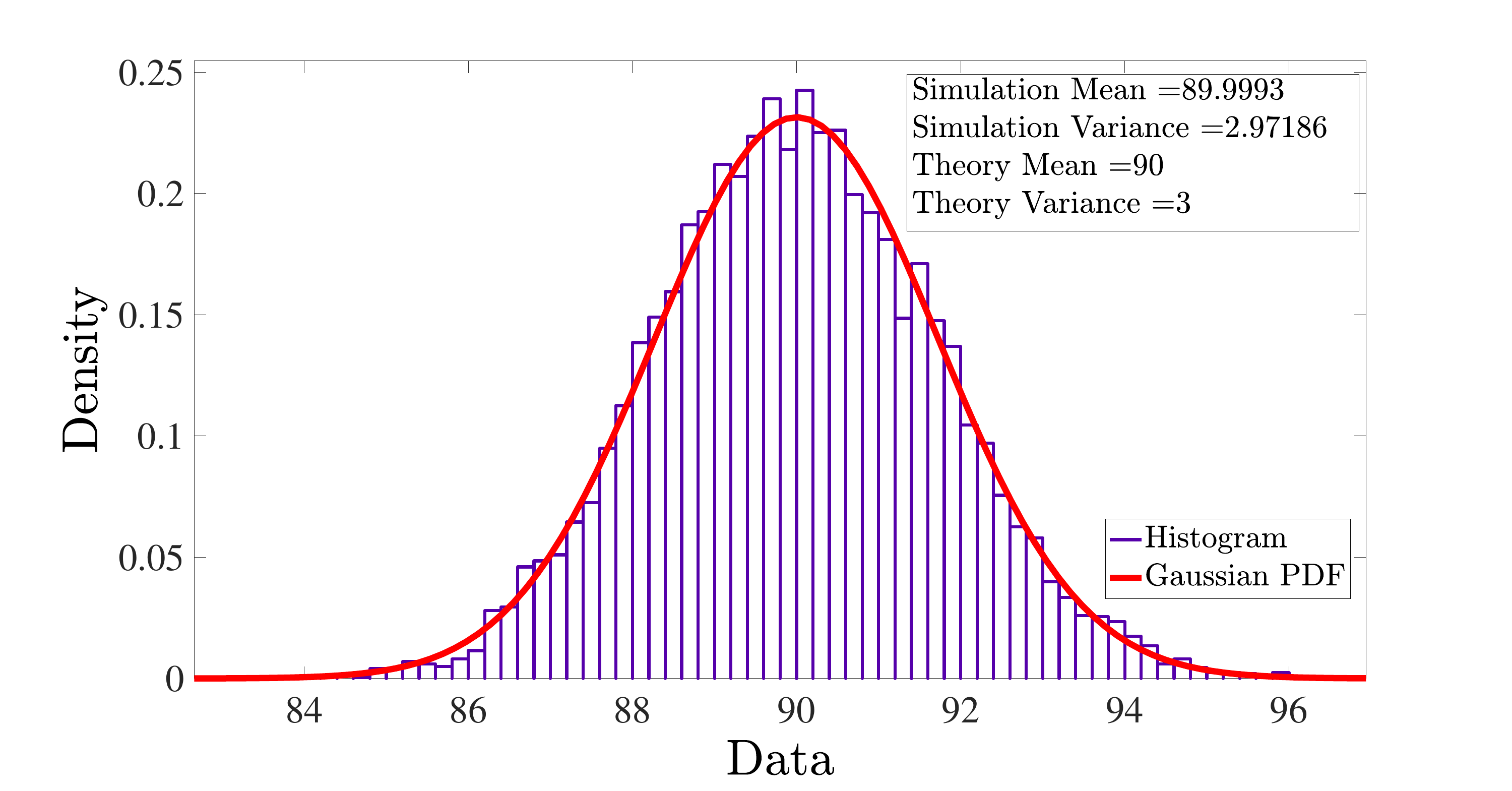}
                \label{CLTFig1}
             \DeclareGraphicsExtensions.
\caption{ The comparison of the numerical and theoretical results for the distribution of the proposed test statistics under hypothesis ${\cal H}_0$ for $M=90$ and $L=30$.}
\end{figure}

\begin{figure}[!t]
\centering          
                \includegraphics[width=.8\columnwidth ,height=2 in]{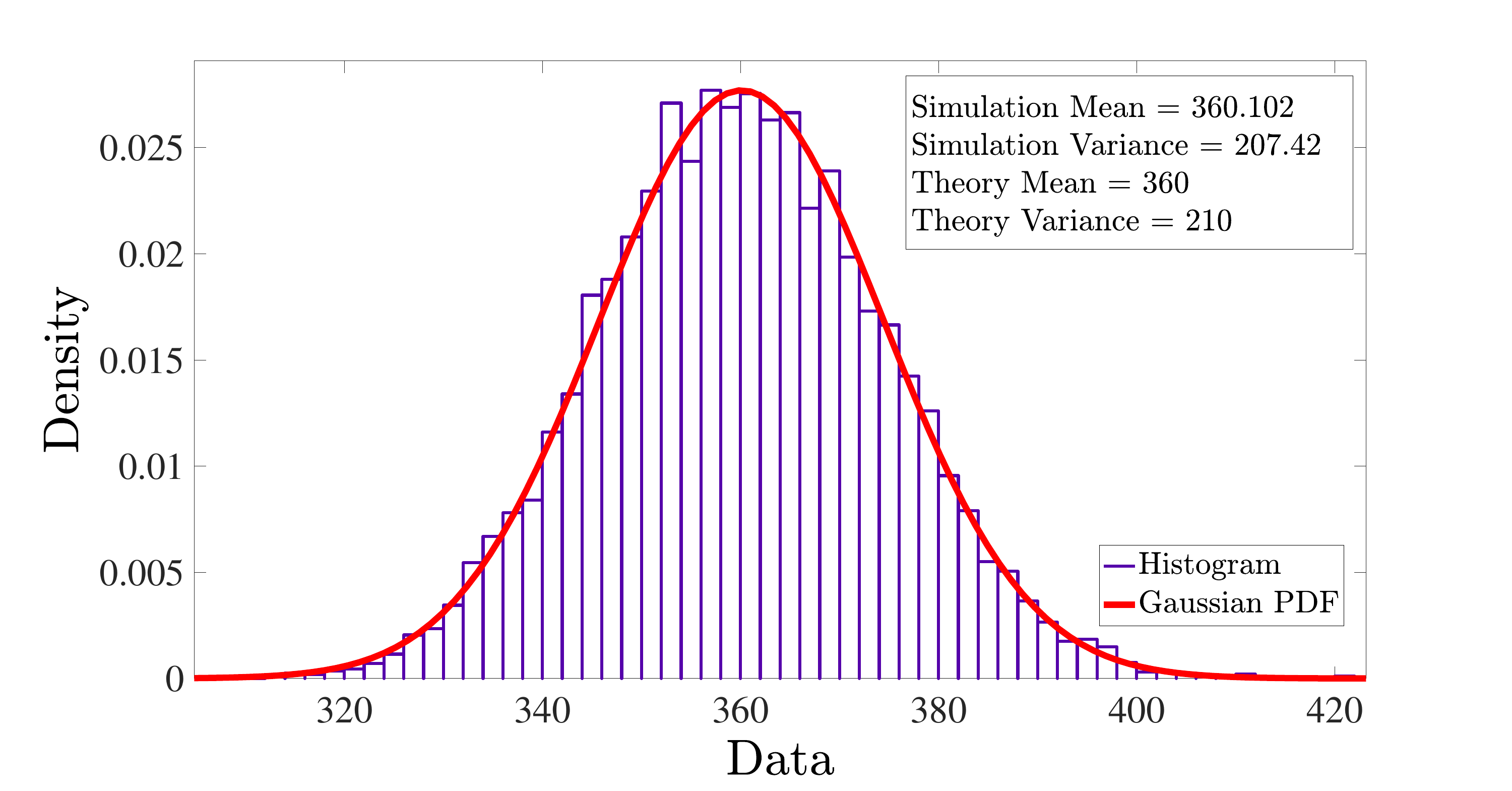}
                \label{CLTFig2} 
         \DeclareGraphicsExtensions.
\caption{ The comparison of the numerical and theoretical results for the distribution of the proposed test statistics under hypothesis ${\cal H}_0$ for $M=90$ and $L=30$.}
\end{figure}       
 
Figure \ref{ROCsample} presents the impact of SNR on the performance of the proposed detectors. For a specific false alarm probability of $P_{\rm fa} = 0.01$, an average SNR of $\overline{\gamma}=-10$ dB, a sample size of $L = 50$, and an antenna count of $M = 30$, all detectors achieve a detection probability exceeding 90\%. Notably, it becomes apparent that the performance improvement diminishes as the number of samples and antennas increases beyond a certain threshold. For instance, for the detector $T_{HDS}$ in Fig. \ref {snrhdjohnsample} by increasing the number of samples from $L = 90$ to $ L =120 $, the probability of detection improves $1dB$ in terms of SNR. Similarly in \ref{snrhdglrd3sample} by increasing the number of samples from $L = 90$ to $L = 120$ the SNR improvement is only $0.3dB$.                 
\begin{figure}[!t]
\centering   
                \includegraphics[width=.9\columnwidth ,height=2 in]{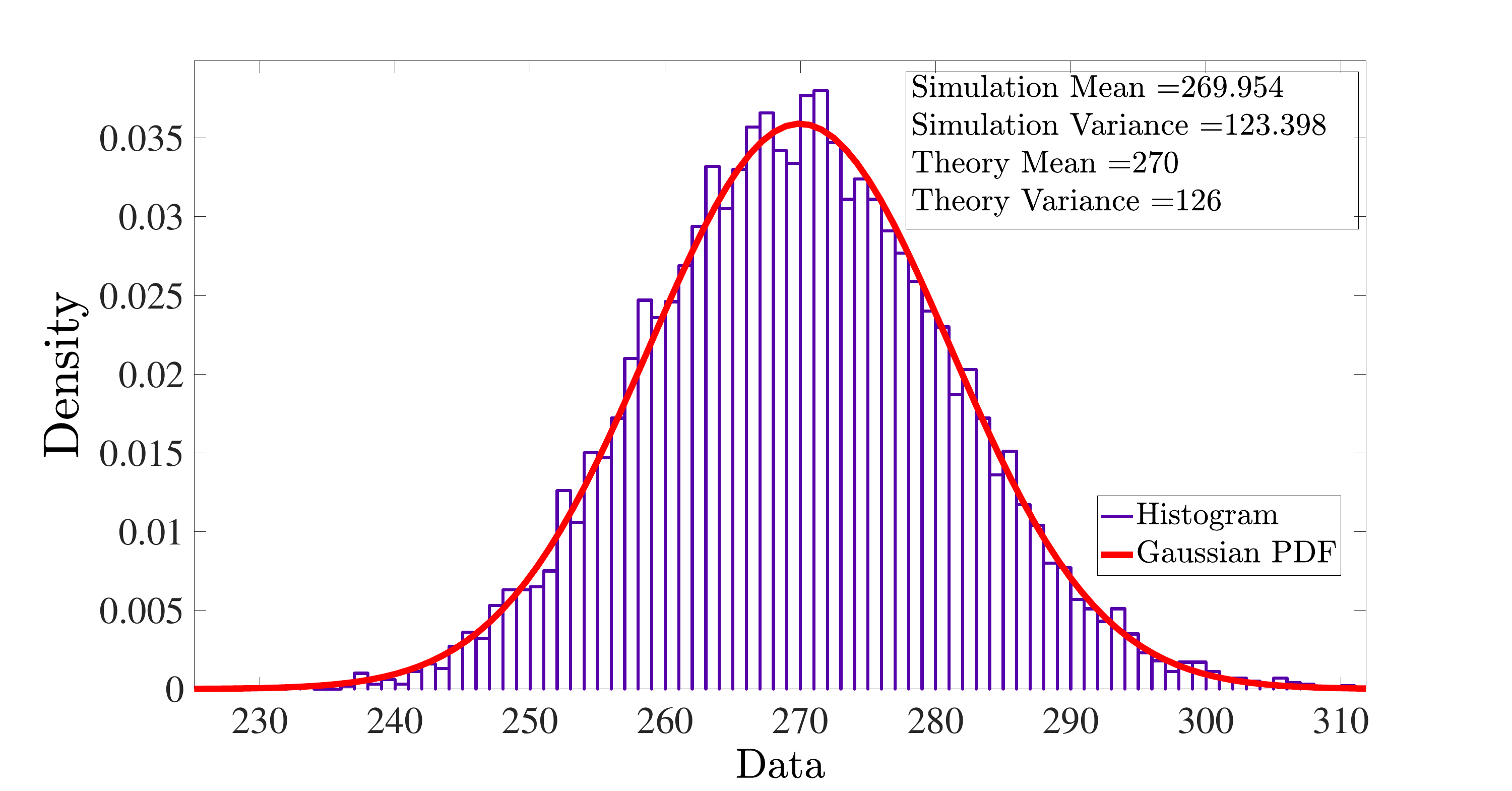}
                \label{CLTFig3}
\DeclareGraphicsExtensions.
\caption{ The comparison of the numerical and theoretical results for the distribution of the proposed test statistics under hypothesis ${\cal H}_0$ for $M=90$ and $L=30$.}
\end{figure}

In conventional cases, when the number of dimensions of the observation matrix is significantly smaller than the number of samples, detectors such as GLR, Rao, and Wald tests follow a chi-square distribution under the null hypothesis. However, when dealing with large array antennas where the number of antennas can exceed the number of samples, this chi-square assumption is no longer valid. As discussed earlier, detectors derived from RMT exhibit Gaussian distributions in such scenarios. In this section, we validate this Gaussian distribution assumption using simulation results.

We plot the histograms of the decision statistics for the $T_{HDL}$, $T_{HDL}$, and $T_{HDL}$ detectors under the hypothesis ${\cal H}_0$, considering different values of $M$ and $L$. Figures 6-8 illustrate these histograms, and it is evident that the means and variances obtained from theoretical calculations closely match the results obtained from Monte-Carlo simulations. Through extensive simulations, we observe that the accuracy of the approximations improves as the ratio of the number of antennas to the number of samples approaches 1.

Another important finding from these figures is that the mean and variance approximations for the $T_{HDL}$ detector are generally more accurate compared to the other two detectors. This highlights the effectiveness of the $T_{HDL}$ detector in capturing the statistical characteristics of the large array antenna scenario.

Overall, the simulation results support the validity of the Gaussian distribution assumption for the derived detectors. The accuracy of the approximations improves with a balanced ratio of antennas to samples, and the $T_{HDL}$ detector exhibits superior performance in terms of mean and variance approximation.
\section{Conclusion}\label{sec:conclusion}
We studied the MASS by using a large array antenna in CR networks. Using RMT, we derived and proposed three detectors based on three particular forms of LSS, which correspond to GLR, Frobenius norm, and Rao versions of conventional MASS detectors. The performances of the proposed detectors under the noise-only hypothesis were investigated by using CLT for random matrices and the mean and variance of related Gaussian distributions computed from the complex integrals. The simulations results showed that the GLR version of a large array antenna has the same form as conventional MASS and, in most practical cases, has similar performance to the proposed Rao test and with a simpler form. Also, it was shown that the calculated Gaussian distribution is quite accurate and can be used to calculate the threshold of decision making for a given false alarm probability. 
\begin{appendices}
\section{Proof of Theorem 3}
\label{app_1}
For this detector $g(z)=z$ and hence from (\ref {phieq}) we have
\begin{align}
\phi_M(\omega)=z_M(\omega)=\omega\big(1-\frac{1}{L}\sum_{m=1}^M\frac{\lambda_m}{\lambda_m-\omega_M}\big)
\end{align}
In this case, the first term in (\ref{main_int}) equals to zero and therefore, we need to calculate the following integral
\small
\begin{align}
 \hat{\eta}=\frac{L}{M}.\frac{1}{2\pi j}\oint_{\mathcal{C}_{\omega}^+}\big(1-\frac{1}{L}\sum_m\frac{\lambda_m}{\lambda_m-\omega}\big)\big(1-\frac{1}{L}\sum_i\frac{\lambda_i^2}{(\lambda_i-\omega)^2}\big)d\omega
 \end{align}
\normalsize
Now by expansion of the terms in the integral and computing them separately, we conclude that only the following term is nonzero
\begin{align}
 \frac{1}{2\pi j}\oint_{\mathcal{C}_{\omega}^+}\frac{1}{L}\sum_m\frac{\lambda_m}{\omega-\lambda_m}d\omega=\frac{1}{L}\sum_m\lambda_m
 \end{align}
 
 and as a result, we will obtain
 \begin{align}
 T_{HDL}=\frac{1}{M}\sum_m\lambda_m=\frac{1}{M}\TRAC(\bm{R})
 \end{align}

\section{Proof of Theorem 4}\label{app_2}

In this case, $g(z)=z^2$ and therefore from (\ref{phieq})
 \begin{align}
\phi_M(\omega)=z^2_M(\omega)=\omega^2\big(1-\frac{1}{L}\sum_{m=1}^M\frac{\lambda_m}{\lambda_m-\omega}\big)^2
\end{align}
Now with regard to (\ref{main_int}), we note that 1st integral is zero because it is the integral of analytic function on a closed contour and thus we need to calculate the following integral
\small
 \begin{align}
 \oint_{\mathcal{C}_{\omega}^+}\omega\big(1-\frac{1}{L}\sum_m\frac{\lambda_m}{\lambda_m-\omega}\big)^2\big(1-\frac{1}{L}\sum_i\frac{\lambda_i^2}{(\lambda_i-\omega)^2}\big)d\omega.
 \end{align}
\normalsize
We expand the terms inside the integral and calculate them separately
\begin{align}
&\omega\big(1+\frac{2}{ L}\sum_m\frac{\lambda_m}{\omega-\lambda_m}+\frac{1}{L^2}\sum\limits_m {\sum\limits_r {\frac{{{\lambda _m}{\lambda _r}}}{{(\omega  - {\lambda _m})(\omega  - {\lambda _r})}}} }\big)\nonumber \\ &\times\big(1-\frac{1}{L}\sum_i\frac{\lambda_i^2}{(\lambda_i-\omega)^2}\big)
 \end{align}
 Now for each of terms, we have
 \begin{align}
 \frac{1}{2\pi j}\oint_{\mathcal{C}_{\omega}^+}\frac{2}{L}\sum_m\frac{\lambda_m\omega}{\omega-\lambda_m}d\omega=\frac{2}{L}\sum_m\lambda_m^2
 \end{align}
  \begin{align}
 \frac{1}{2\pi j}\oint_{\mathcal{C}_{\omega}^+}-\frac{1}{L}\sum_i\frac{\lambda_i^2\omega}{(\omega-\lambda_i)^2}d\omega=-\frac{1}{L}\sum_m\lambda_i^2
 \end{align}
 \begin{align}
&\frac{1}{2\pi j} \oint_{\mathcal{C}_{\omega}^+}\frac{1}{L^2}\sum\limits_m {\sum\limits_r {\frac{{{\lambda _m}{\lambda _r}\omega}}{{(\omega  - {\lambda _m})(\omega  - {\lambda _r})}}} }d\omega \nonumber \\ &=\frac{1}{L^2}\sum\limits_m {\sum\limits_r {\frac{{{\lambda _m^2}{\lambda _r}}}{{(\lambda_m  - {\lambda _r})}}} }+\frac{1}{L^2}\sum\limits_m {\sum\limits_r {\frac{{{\lambda _m}{\lambda _r^2}}}{{(\lambda_r  - {\lambda _m})}}} }
\nonumber\\&=\frac{1}{L^2}\sum\limits_m {\sum\limits_r {{\lambda _m}{\lambda _r}} } 
\end{align}
 \begin{align}
\frac{1}{2\pi j}\oint_{\mathcal{C}_{\omega}^+}- \frac{2}{{{L^2}}}\sum\limits_m {\sum\limits_i {\frac{{{\lambda _m}\lambda _i^2\omega}}{{(\omega-{\lambda _m} ){{({\lambda _i} - \omega )}^2}}}} }d\omega=0
\label{int2}
\end{align}

The last term is zero since $ \lambda_i $ is a 2nd order pole is for one special value of $\lambda_i $, the residue equals to
\begin{align}
\frac{\mathrm{d}}{\mathrm{d\omega}}\big(\frac{\lambda_m\lambda_i^2\omega}{\omega-\lambda_m}\big)\big|_{\omega=\lambda_i}=&\frac{\lambda_m\lambda_i^2(\omega-\lambda_m)-\lambda_m\lambda_i^2\omega}{(\omega-\lambda_m)^2}\big|_{\omega=\lambda_i}\nonumber\\=&\frac{-\lambda_m^2\lambda_i^2}{(\lambda_i-\lambda_M)^2}
\label{martabe22}
\end{align}
and similarly because $\lambda_m $ is a 1st order pole, for one special $\lambda_m $, the residue is
\begin{align}
 \frac{\lambda_m\lambda_i^2\omega}{(\lambda_i-\omega)^2}\big|_{\lambda_m}=\frac{\lambda_m^2\lambda_i^2}{(\lambda_i-\lambda_m)^2}
\label{martabe11}
\end{align}

and as these expressions are equal with different signs, according to the integral residue theorem, it will be zero, and thus the final result is

 \begin{align}
 T_{HDS}=\frac{1}{M}\TRAC(\bm{R}^2)+\frac{1}{ML}(\TRAC(\bm{R}))^2
 \end{align}

\section{Proof of Theorem 5}\label{app_3}
For the case of quadratic function as $g(z)=(z-1)^2$, from (\ref {phieq}), we can write
\begin{align}
\phi_M(z)=(z-1)^2=\omega^2\big(\frac{\omega-1}{\omega}-\frac{1}{L}\sum_{i=1}^M\frac{\lambda_i}{\lambda_i-\omega}\big)^2
\end{align}
First term in (\ref {main_int}) is zero since $g(z)$ is an analytical function and thus, we need to calculate the following integral
\small 
\begin{align}
\frac{1}{2\pi j}\oint_{\mathcal{C}_{\omega}^+}\omega\big(\frac{\omega-1}{\omega}-\frac{1}{L}\sum_m\frac{\lambda_m}{\lambda_m-\omega}\big)^2\big(1-\frac{1}{L}\sum_i\frac{\lambda_i^2}{(\lambda_i-\omega)^2}\big)d\omega,
\end{align}
\normalsize
 where the inside the integral can be expanded as
\small
\begin{align}
&\omega\big(\frac{(\omega-1)^2}{\omega^2}+\frac{2(\omega-1)}{\omega L}\sum_m\frac{\lambda_m}{\omega-\lambda_m}+\nonumber\\&\frac{1}{L^2}\sum\limits_m {\sum\limits_r {\frac{{{\lambda _m}{\lambda _r}}}{{(\omega  - {\lambda _m})(\omega  - {\lambda _r})}}} }\big) \big(1-\frac{1}{L}\sum_i\frac{\lambda_i^2}{(\lambda_i-\omega)^2}\big)
\end{align}
\normalsize
which to solve the integral, below, we calculate each of them separately 

\begin{align}
\oint_{\mathcal{C}_{\omega}^+}\frac{(\omega-1)^2}{\omega}d\omega=0,
\end{align}
because of Cauchy integral theorem,
\small
\begin{align}
\frac{1}{2\pi j}\oint_{\mathcal{C}_{\omega}^+}\frac{2(\omega-1)}{L}\sum_m\frac{\lambda_m}{\omega-\lambda_m}d\omega=\frac{2}{L}\sum_m\lambda_m(\lambda_m-1)
\end{align}
\normalsize
\begin{align}
&\frac{1}{2\pi j}\oint_{\mathcal{C}_{\omega}^+}\frac{1}{L^2}\sum\limits_m {\sum\limits_r {\frac{{{\lambda _m}{\lambda _r}\omega}}{{(\omega  - {\lambda _m})(\omega  - {\lambda _r})}}} }d\omega \nonumber\\&=\frac{1}{L^2}\sum\limits_m {\sum\limits_r {\frac{{{\lambda _m^2}{\lambda _r}}}{{(\lambda_m  - {\lambda _r})}}} }+\frac{1}{L^2}\sum\limits_m {\sum\limits_r {\frac{{{\lambda _m}{\lambda _r^2}}}{{(\lambda_r  - {\lambda _m})}}} }
\nonumber\\&=\frac{1}{L^2}\sum\limits_m {\sum\limits_r {{\lambda _m}{\lambda _r}} } 
\label{trace2}
\end{align}
\begin{align}
\frac{1}{2\pi j}\oint_{\mathcal{C}_{\omega}^+}- \frac{1}{L}\sum\limits_i {\frac{{{\lambda _i}^2{{(\omega  - 1)}^2}}}{{\omega {{(\omega  - {\lambda _i})}^2}}}}d\omega=-\frac{1}{L}\sum\limits_i{\lambda _i^2}
\label{int1}
\end{align}
Since the $\omega $ a 2nd order pole and using the integral formula of residues for particular $\omega_i $ we have
\begin{align}
\frac{\mathrm{d}}{\mathrm{d\omega}}\big(\frac{\lambda_i^2(\omega-1)^2}{\omega}\big)\big|_{\omega=\lambda_i}
=&\frac{2\lambda_i^2(\omega-1)\omega-\lambda_i^2(\omega-1)^2}{\omega^2}\big|_{\omega=\lambda_i}\nonumber\\=&\lambda_i^2-1
\end{align}
\begin{align}
\oint_{\mathcal{C}_{\omega}^+}- \frac{2}{{{L^2}}}\sum\limits_m {\sum\limits_i {\frac{{{\lambda _m}\lambda _i^2(\omega  - 1)}}{{(\omega-{\lambda _m} ){{({\lambda _i} - \omega )}^2}}}} }d\omega=0.
\label{int2}
\end{align}

Because $\lambda_i $ is a 2nd order pole is for a given $\lambda_i $, the residue equals to
\small
\begin{align}
\frac{\mathrm{d}}{\mathrm{d\omega}}\big(\frac{\lambda_m\lambda_i^2(\omega-1)}{\omega-\lambda_m}\big)\big|_{\omega=\lambda_i}=&\frac{\lambda_m\lambda_i^2(\omega-\lambda_m)-\lambda_m\lambda_i^2(\omega-1)}{(\omega-\lambda_m)^2}\big|_{\omega=\lambda_i}\nonumber\\=&\frac{\lambda_m\lambda_i^2-\lambda_m^2\lambda_i^2}{(\lambda_i-\lambda_M)^2}
\label{martabe2}
\end{align}
\normalsize
similarly for $ \lambda_m $ as a 1st order pole, the residue is equal to
\begin{align}
\frac{\lambda_m\lambda_i^2(\omega-1)}{(\lambda_i-\omega)^2}\big|_{\lambda_m}=\frac{\lambda_m\lambda_i^2(\lambda_m-1)}{(\lambda_i-\lambda_m)^2}
\label{martabe1}
\end{align}

which~(\ref {martabe1}) is exactly~(\ref {martabe2}) with different sign hence by the total integral becomes zero. Finally, we notice that
\small
\begin{align}
\oint_{\mathcal{C}_{\omega}^+}- \frac{1}{{{L^3}}}\sum\limits_m {\sum\limits_r {\sum\limits_i {\frac{{{\lambda _m}{\lambda _r}\lambda _i^2\omega }}{{(\omega  - {\lambda _m})(\omega  - {\lambda _r}){{({\lambda _i} - \omega )}^2}}}} } }d\omega=0
\end{align}
\normalsize
This is because $ \lambda_m $ is a 1st order pole and residue is equal to
\begin{align}
{\frac{{{\lambda _m}{\lambda _r}\lambda _i^2\omega }}{{(\omega  - {\lambda _r}){{({\lambda _i} - \omega )}^2}}}}\big|_{\omega=\lambda_m}={\frac{{{\lambda _m}^2\lambda _i^2{\lambda _r}}}{{({\lambda _m} - {\lambda _r}){{({\lambda _i} - {\lambda _m})}^2}}}}
\end{align}

The same way for $ \lambda_r $ as a 1st order poles, we obtain the residue as
\begin{align}
{\frac{{{\lambda _m}{\lambda _r}\lambda _i^2\omega }}{{(\omega  - {\lambda _m}){{({\lambda _i} - \omega )}^2}}}}\big|_{\omega=\lambda_r}={\frac{{{\lambda _r}^2\lambda _i^2{\lambda _m}}}{{({\lambda _r} - {\lambda _m}){{({\lambda _i} - {\lambda _r})}^2}}}}
\end{align}
and  $ \lambda_i $ as a 2nd order pole,
\begin{align}
&\frac{\mathrm{d}}{\mathrm{d\omega}}\big({\frac{{{\lambda _m}{\lambda _r}\lambda _i^2\omega }}{{(\omega  - {\lambda _m})(\omega  - {\lambda _r})}}}\big)\big|_{\omega=\lambda_i}\\
&=\frac{\lambda _m\lambda _r\lambda _i^2(\lambda _i-\lambda _m)(\lambda _i-\lambda _r)-\lambda _m\lambda _r\lambda _i^3(2\lambda _i-\lambda _m-\lambda _r)}{(\lambda _i-\lambda _m)^2(\lambda _i-\lambda _r)^2}\nonumber
\end{align}

which using all these results and after mathematical simplification, we get 
\begin{align}
&(\lambda _i^4\lambda _r\lambda_m-\lambda _m^2\lambda _i^2\lambda _r^2)(\lambda_m-\lambda_r)\nonumber\\&+(\lambda _m^2\lambda _i^2\lambda _r^2-\lambda _i^4\lambda _r\lambda_m)(\lambda_m-\lambda_r)=0
\end{align}
Consequently, given the above results, we have
\begin{align}
T_{HDQ}&=\frac{1}{M}\sum\limits_m{\lambda _m^2}+\frac{1}{ML}\sum\limits_m {\sum\limits_r {{\lambda _m}{\lambda _r}} }-\frac{2}{M}\sum\limits_m{\lambda _m}\nonumber
\\&=\frac{1}{M}\TRAC(\bm{R}^2)+\frac{1}{ML}(\TRAC(\bm{R}))^2-\frac{2}{M}\TRAC(\bm{R})
\end{align}
\section{Proof of Theorem }\label{App-B}
\subsection {Mean Calculation}
As we assumed a complex system model, according to (\ref{MEAN}), $\mu$ will be zero. However, form \cite{Bai2009}, we know that the average has another correction term which equals to $ MF^c(g)$, where $F^c(g) $ is related to Marcenko-Pastur distribution and calculated as
\begin{align}
{F^c}(g) = \int_{a(c)}^{b(c)} {\frac{{g(z)}}{{2\pi cz}}\sqrt {(b(c) - z)(z - a(c))} dz},
\end{align}
where, $a (c) = 1-2 \sqrt {c} + c $ and
$ b (c) = 1 + 2\sqrt {c} + c $. By defining the auxiliary variable $\zeta$ as $z=1+c-2\sqrt{c}~{cos \zeta}$, then $a(c)$ and $b(c)$ will correspond to $\zeta = 0 $ and $ \zeta = \pi $ and for the square-root term we will have
\small
\begin{align}
\sqrt{(b(c)-1-c+2\sqrt{2}cos\zeta)(1+c-2\sqrt{c}cos\zeta-a(c))}=4c sin^2\zeta.
\end{align}
\normalsize

We first assume that $g(z) = (z-1)^ 2$ and therefore
\begin{align}
F^c(g)=\int_0^\pi\frac{(c-2\sqrt{c}cos\zeta)^2 4c sin^2\zeta }{2\pi c(1+c-2\sqrt{c} cos\zeta)}d\zeta
\end{align}

Now we use change the variable as $\alpha=-\frac{1+c}{2\sqrt{c}}$ and hence $1+c-2\sqrt{c} cos\zeta=-2\sqrt{c}(cos\zeta+\alpha)$, after some mathematical manipulations we can conclude that $F^c(g)=c$ and therefore the mean equals to
$ Mc $.

If we denote $k$th moments for the Marcenko-Pastur distribution by $ \beta_k$ , then from \cite{Baibook}, we have
\begin{align}
\label{pasturemomen}
\beta_k=&\frac{1}{2\pi c}\int_{a(c)}^{b(c)}x^{k-1}\sqrt{(b-x)(x-a)}dx\\\nonumber=&\sum_{r=0}^{k-1}\frac{1}{r+1}{k\choose r}{k-1\choose r}c^r
\end{align}

Now for $g(z) = z$ and $g(z)=z^2$, we notice that the expression in (\ref{pasturemomen}) equals to $\beta_1=1$ and $\beta_2=c+1$, respectively, and as a result, the means will be equal to $M$, $M(c + 1)$.

\subsection{Variance}
First, we note that in (\ref {inversem}), if we assume $\underline {m} (z) $  is the Stieltjes transform of $\underline {F} ^ c = (1-c) \bm {1} _ {(0, \infty)} + cF ^ c$, we can obtain
  \begin {align}
  z = - \frac {1} {\underline {m}} + \frac {c} {1+ \underline {m}}
  \label {change}
  \end {align}

Now, to calculate the variance in  (\ref {COV}) for the case of $g(z)=(z-1)^2$, we consider $g(z_1) = (z_1-1) ^ 2$ and $g(z_2) = (z_2-1) ^ 2 $ which their product equals to
\begin{align}
g(z_1)g(z_2)&=z_1^2z_2^2-2z_1^2z_2-2z_1z_2^2\nonumber\\&+4z_1z_2+z_1^2+z_2^2-2z_1-2z_2+1
\label{bast}
\end{align}
and we calculate each of the terms separately. For
$\oint\oint {\frac{{{z_1}}}{{(m{}_1 - {m_2})}}} d{m_1}d{m_2}$
by substituting $z_1 $ in (\ref {change}) and given that the integral contour only contains only  $-1$ as a pole and assumes
$ m_2 $ fixed, we get 
\small
\begin{align}
\oint\oint {\frac{{{z_1}}}{{(m{}_1 - {m_2})}}} d{m_1}d{m_2}
= 2\pi j\oint\frac{c}{(1+m_2)^2}dm_2=0
\end{align}
\normalsize
Similarly,
\begin{align}
&\oint\oint {\frac{{{z_1}^2}}{{(m{}_1 - {m_2})^2}}} d{m_1}dm_2=
\oint\big(\frac{2(c-1)^2}{(1+m_2)^3}\nonumber\\&+\frac{2(c-1)(3+m_2)}{(1+m_2)^3}+\frac{2(2+m_2)}{(1+m_2)^3}\big)dm_2=0
\end{align}

Therefore, given the above results $v (z_1 ^ 2-2z_1,1) = 0 $ and similarly, $ v (1, z_2 ^ 2-2z_2) = 0 $. 

Now we have to calculate the following expressions
\small
\begin{align}
v(g)=v(z_1^2,z_2^2)-2v(z_1^2,z_2)-2v(z_1,z_2^2)+4v(z_1,z_2)
\end{align}
\normalsize
which we have:
\begin{align}
v(z_1,z_2)&=\oint {\oint {\frac{{{z_1}{z_2}}}{{{{({m_1} - {m_2})}^2}}}d{m_1}} } d{m_2} \nonumber\\
&=\oint z_2\oint \frac{z_1}{(m_1-m_2)^2}dm_1dm_2\\\nonumber&=\oint\frac{c z_2}{(1+m_2)^2}dm_2=\oint\frac{c(c-1)m_2-c}{m_2(1+m_2)^3}dm_2=c
\end{align}
\begin{align}
v(z_1^2,z_2)&=\oint {\oint {\frac{{{z_1}^2{z_2}}}{{{{({m_1} - {m_2})}^2}}}d{m_1}} } d{m_2}\nonumber\\
&=\oint z_2\oint \frac{z_1^2}{(m_1-m_2)^2}dm_1dm_2\nonumber \\\nonumber&=\oint\big( \frac{2 z_2 c}{(1+m_2)^2}+\frac{2 c^2 z_2}{(1+m_2)^3}\big)dm_2\\\nonumber&=\oint\big( \frac{2 c(c-1)m_2-2c}{m_2(1+m_2)^3}+\frac{2 c^2 (c-1)m_2-2c^2}{m_2(1+m_2)^4}\big)dm_2\nonumber\\&=2c(c+1)
\end{align}
\begin{align}
v(z_1^2,z_2)=&\oint {\oint {\frac{{{z_1}^2{z_2}^2}}{{{{({m_1} - {m_2})}^2}}}d{m_1}} } d{m_2}\nonumber\\
&=\oint z_2^2\oint \frac{z_1^2}{(m_1-m_2)^2}dm_1dm_2\nonumber\\\nonumber&=\oint\big( \frac{2 z_2^2 c}{(1+m_2)^2}+\frac{2 c^2 z_2^2}{(1+m_2)^3}\big)dm_2\\\nonumber&=\oint\big( \frac{2 c(c-1)^2m_2^2-4c(c-1)m_2+2c}{m_2^2(1+m_2)^4}dm_2\\\nonumber &+\oint\frac{2 c^2(c-1)^2m_2^2-4c^2(c-1)m_2+2c^2}{m_2^2(1+m_2)^5}\big)dm_2\\&=4c^3+10c^2+4c
\end{align}
and finally as a result
$v(g) $ is equal to
\small
\begin{align}
v(g)=4c^3+10c^2+4c-8c^2-8c+4c=4c^3+2c^2=2c^2(2c+1)
\end{align}
\normalsize
It should be noted that in the process of calculating the variance for the function
$ g(z) = z^2-1$, variance of functions of $g(z) = z$ and $g(z) = z^2$ are also calculated.

\end{appendices}

\bibliographystyle{IEEEtran}
\bibliographystyle{IEEEtran}

\end{document}